\documentclass[12pt]{article}
\usepackage[utf8]{inputenc}

\usepackage{geometry}
\usepackage{amsmath}
\allowdisplaybreaks
\usepackage{mathtools}
\usepackage{blkarray}
\usepackage{amssymb}
\usepackage{amsthm}
\usepackage[shortlabels]{enumitem}
\usepackage{setspace}
\usepackage[colorlinks,linkcolor=blue,citecolor=blue,urlcolor=blue,bookmarks=false,hypertexnames=true]{hyperref} 
\usepackage{xcolor}
\usepackage{url}
\usepackage{etoolbox}

\usepackage[round]{natbib}
\newtheorem{theorem}{Theorem}
\newtheorem{theorem-app}{Theorem}[section]
\newtheorem{definition}{Definition}
\newtheorem{definition-app}{Definition}[section]
\newtheorem{axiom}{Axiom}
\newtheorem{axiom-app}{Axiom}[section]

\newtheorem{proposition}{Proposition}

\newtheorem{lemma-app}{Lemma}[section]
\newtheorem{corollary}{Corollary}
\newtheorem{example}{Example}
\newtheorem{example-app}{Example}[section]

\newtheorem{assumption}{Assumption}

\newtheorem{observation}{Observation}

\newcommand{\R}{\mathbb{R}}

\newcommand{\cI}{\mathcal{I}}

\newcommand{\CVLAM}{\underline{V}_{\mathrm{LAM},\theta_0}}
\newcommand{\AVLAM}{\overline{V}_{\mathrm{LAM},\theta_0}}

\makeatletter
\def\@fnsymbol#1{\ensuremath{\ifcase#1\or \dagger\or \ddagger\or
   \mathsection\or \mathparagraph\or \|\or **\or \dagger\dagger
   \or \ddagger\ddagger \else\@ctrerr\fi}}
    \makeatother

\title{Axiomatizing Local Asymptotic Minimax Risk}
\author{\makebox[.25\linewidth]{Ricky Li\thanks{This version: September 17, 2026. MIT Department of Economics; email: \protect\texttt{rickyli@mit.edu}. \\~\\
We are especially grateful to Isaiah Andrews, Anna Mikusheva, and Drew Fudenberg for their insightful feedback. We also thank participants of the MIT Theory and Econometrics lunch seminars for helpful comments.
}} \\
{MIT}}
\date{}

\begin{document}

\maketitle

\begin{abstract}
Local asymptotic minimax (LAM) risk is a foundational efficiency criterion in statistics and econometrics. The literature uses two definitions of LAM risk: one which appears in classical lower bounds and another which appears in arguments establishing attainment of those bounds. Conventional efficiency arguments are consistent with \emph{any} estimator-dependent weighted average of the two, and consequently do not reveal which of these \emph{generalized $\alpha$-LAM} risk indices describes researchers' actual preferences. We take a decision-theoretic approach to systematically resolve this ambiguity. We axiomatically characterize the set of preferences with a generalized $\alpha$-LAM representation, and we document that such preferences may violate basic rationality requirements such as Monotonicity. Motivated by this, we axiomatically characterize the subset with constant-weight representations. Within this class, a \emph{Sample Uncertainty Aversion} axiom uniquely selects attainment LAM. We argue that Sample Uncertainty Aversion is normatively appealing, and we therefore recommend attainment LAM as the functional form for LAM risk.
\end{abstract} 

\clearpage

\section{Introduction}
\label{sec:intro}
The \emph{local asymptotic minimax} (LAM) risk criterion is a foundational notion of efficiency in parametric and semiparametric statistics \citep{hajek1972local, bickel1993efficient, van2000asymptotic}. Econometricians use it to evaluate efficient estimation in structural models \citep{hirano2003asymptotic}, efficient estimation of non-smooth functionals \citep{song2014local,fang2014optimal}, efficient treatment assignment rules \citep{hirano2009asymptotics}, and efficient policies in sequential decision problems \citep{hirano2023asymptotic,adusumilli2025risk}. Despite this central role, existing motivations for adopting LAM are heuristic. A standard motivation is that localizing the parameter at the rate at which sampling uncertainty vanishes ensures that the statistical difficulty of the estimation problem remains non-negligible in large samples, thereby yielding more useful comparisons than fixed-parameter asymptotics.\footnote{See Sections 7.3 and 8.3 of \cite{van2000asymptotic} and Section 2.5.2 of \cite{hirano2020asymptotic}. LAM also reveals superefficient procedures' poor performance near points of superefficiency (e.g., Hodges' estimator in \citealt{lehmann1998theory} pg. 442 and the ``oracle property'' in \citealt{leeb2008sparse}).} Although this intuitively motivates a \emph{local asymptotic} perspective towards optimality, it does not uniquely recommend the LAM criterion, nor does it specify a functional form for LAM risk. This note aims to fill these gaps using a decision-theoretic approach.

A key starting point for our analysis is that, while finite-sample optimality criteria evaluate estimators, asymptotic optimality criteria evaluate \emph{estimator sequences} indexed by sample size. Consequently, such criteria necessarily take a stance on how to trade off risk across sample sizes, even though the researcher's actual estimation problem features a fixed sample size. This is a convenient fiction for asymptotic statistics, but it poses a challenge for understanding the LAM criterion, since existing motivations do not fully specify how to aggregate risk across sample sizes.\footnote{They also do not explain taking the \emph{worst-case} risk over local parameters for a fixed sample size. A naive motivation is to view LAM as a limiting analog of finite-sample minimaxity and point to previous decision-theoretic work studying the latter \citep{stoye2012new,gilboa1989maxmin}. However, ``[LAM] optimality$\ldots$ is conceptually different from finite sample minmaxity, and neither$\ldots$ implies the other'' (\citealt{hirano2020asymptotic} pg. 311). Our main axiomatic results make this ``conceptual difference'' precise.} Indeed, the literature uses two functional forms for LAM risk: classical lower bounds \citep{hajek1972local,le1979theorem} apply to the \emph{limiting best-case} local risk over sample size subsequences ($\liminf$), while subsequent arguments which establish attainment of these bounds \citep{van2002semiparametric,song2014local} use the corresponding \emph{limiting worst-case} ($\limsup$). Observation \ref{obs:gen-alpha-LAM-iff-obv-cons} shows that standard efficiency arguments are \emph{transparently consistent} with any estimator-dependent weighted average of these two endpoints. It is therefore unclear which functional form for LAM risk from this class of \emph{generalized $\alpha$-LAM} risk indices (if any) describes researchers' actual preferences.

Proposition~\ref{prop:stable-local-risks} and Corollary~\ref{cor:regular-stable-losses} provide one reason why this issue may not have been formally addressed: all generalized $\alpha$-LAM risk indices coincide for estimator sequences with limiting local estimation error laws and asymptotically uniformly integrable local losses, which are standard regularity conditions. Consequently, observing the researcher's rankings of such estimator sequences cannot distinguish among this class of functional forms. Nevertheless, we exhibit a simple estimator sequence (Example \ref{ex:gaussian}) which does not satisfy these conditions, such that the difference among generalized $\alpha$-LAM risk indices is maximally stark. Given that a commonly stated appeal of the LAM criterion is that it evaluates \emph{any} estimator sequence,\footnote{See, e.g., \cite{van2000asymptotic} Section 8.7 and \cite{van2002semiparametric} pg. 348.} our research question remains relevant: which index describes researchers' preferences?

We use the toolkit of microeconomic decision theory to systematically answer this question. For each class of LAM risk indices of interest, we offer a set of conditions (called \emph{axioms}) that are necessary and sufficient for the researcher's rankings to be consistent with some index in the class. Our axioms specify how the researcher aggregates risk across sample sizes and local parameters. Since the researcher ranks \emph{sequences} indexed by sample size, we interpret the decision environment as one in which the researcher forms their rankings \emph{before} learning the sample size of their dataset, and thus faces \emph{ex-ante uncertainty} about the sample size (as well as the local parameter). This is the decision-theoretic interpretation of the convenient fiction described above that is inherent in any asymptotic optimality criterion. 

With this interpretation in mind, we briefly summarize and interpret our main results. As an initial benchmark, Theorem~\ref{thm-main-result} axiomatically characterizes the class of generalized $\alpha$-LAM preferences, thereby providing a complete understanding of the behavioral content of merely being (transparently) consistent with standard efficiency arguments. With this result in hand, we document that generalized $\alpha$-LAM preferences are sufficiently broad to allow for violations of basic rationality requirements like \emph{Monotonicity}. Our examples suggest that these violations are due to unrestricted dependence of the weights on the procedure being evaluated.

To address this limitation, we focus on the subset of \emph{$\alpha$-LAM} preferences represented by a \emph{fixed constant}-weighted average of classical and attainment LAM. Theorem~\ref{thm:fixed-alpha-LAM} restores the basic rationality restrictions, in addition to other axioms, to axiomatically characterize $\alpha$-LAM. Within this subclass, axioms governing the researcher's attitudes towards \emph{sample size uncertainty} play a key role: \emph{Sample Uncertainty Aversion} (which requires a \emph{preference for hedging}) uniquely selects attainment LAM, while \emph{Sample Uncertainty Affinity} (which requires an \emph{aversion to hedging}) uniquely selects classical LAM (Corollary \ref{cor:CLAM-and-ALAM}). This lays bare the difference between risk indices in terms of a concrete and intuitive consequence for researcher behavior. 

Finally, we argue that Sample Uncertainty Aversion is normatively appealing. Consequently, out of the candidate LAM risk representations we study, we recommend attainment LAM. Our axiomatic characterization of attainment LAM therefore complements and refines the heuristic motivations for the LAM criterion described above.

\paragraph{Related literature and key contributions.} 
Providing decision-theoretic foundations for notions of statistical optimality has a long tradition \citep{wald1950statistical,hurwicz1951optimality,savage1954foundations}. More recent work has studied finite-sample statistical optimality criteria from an axiomatic perspective. \cite{stoye2011statistical} draws upon previous microeconomic decision theory work to characterize Bayes, minimaxity, minimax regret, and admissibility, as well as the \emph{Hurwicz} or \emph{$\alpha$-maxmin} criterion which mixes between the best- and worst-case according to a fixed weight, while \cite{stoye2012new} studies these and other finite-sample criteria in a decision environment consisting of risk functions. \cite{andrews2026misspecification} provide axioms which characterize various classes of misspecification-averse criteria used in the literature; we adapt some of their technical tools to elicit and aggregate conditional preferences. In contrast to the finite-sample axiomatic analysis of these papers, our focus is on \emph{local asymptotic} optimality criteria. Rather than studying an established class of representation functional forms, we begin with the largest class that is in a precise sense transparently consistent with standard efficiency arguments in econometrics, and we use our axiomatic analysis to make a principled recommendation for a functional form for LAM risk.

Many of our axioms and results are adapted from or build upon related axioms and representation results in the microeconomic decision theory literature. Key building blocks include the Consistency and Caution axioms \citep{gilboa2010objective,cerreia2024making}, maxmin expected utility \citep{gilboa1989maxmin}, and the complete patience representation studied in \cite{marinacci1998axiomatic} Theorem 7. The axiomatic results in this note contribute to the literature which studies $\alpha$-maxmin and generalized $\alpha$-maxmin representations \citep{arrow1972optimality,ghirardato2004differentiating,cerreia2011rational}. In particular, \cite{chateauneuf2007choice} study a class of representations which mixes between expected utility, the best-case, and the worst-case. The special case of their subclass which puts zero weight on expected utility (their \emph{Hurwicz capacity}) is analogous to our $\alpha$-LAM preferences restricted to a particular subdomain of our choice environment, although they do not axiomatize this case. \cite{drugeon2023alpha} study a class of asymptotic $\alpha$-maxmin criteria in a choice environment of intertemporal utility streams. Their analysis takes as primitive two preference relations, which they interpret as the preferences of society for the near and far future.

Relative to these papers, we view our key contributions as follows. First, we define and axiomatically characterize a collection of asymptotic Hurwicz-type criteria whose functional forms are tightly connected to standard efficiency arguments in statistics and econometrics, and we use our axioms to provide a unique recommendation. Given that asymptotic efficiency is the benchmark notion of optimality among practitioners in these fields, we view this as a compelling use case of the microeconomic decision theory toolkit. Second, our axioms make transparent the role that attitudes towards sample size risk aggregation play in adjudicating among different candidate LAM risk indices. On a technical level, we streamline the axiomatic analysis by taking only one preference as primitive---the researcher's actual ranking of procedures---and elicit other preferences of interest from the primitive preference. The multidimensional nature of uncertainty in our setting (in which outcomes jointly depend on the sample size and local parameter) requires specialized techniques which we hope are useful for future work in decision-making under multidimensional uncertainty. 

\paragraph{Roadmap.}
The rest of this note proceeds as follows. Section~\ref{sec:setup} sets up the standard parametric point estimation environment and shows that conventional efficiency arguments do not pin down a functional form for LAM risk. Section~\ref{sec:axioms} states and interprets our axioms, characterizes generalized $\alpha$-LAM and $\alpha$-LAM preferences, and then characterizes classical and attainment LAM. Appendix~\ref{app:LAN} provides a brief exposition of local asymptotic normality for non-familiar readers. Appendix~\ref{app-sec-2-proofs} proves the results in Section~\ref{sec:setup}. Appendix~\ref{app:AA} develops axiomatic foundations for the main text's decision environment in a standard Anscombe--Aumann setup. Appendix~\ref{app:util-act-lemmas} proves the axiomatic representation results in Section~\ref{sec:axioms}.

\section{Functional forms for LAM risk}
\label{sec:setup}

\paragraph{Statistical decision environment.} We adapt the standard parametric point estimation setup of \cite{van2000asymptotic} Chapter 8. Let $\Theta\subseteq \mathbb{R}^k$ be an open set of parameters, and let $\psi: \Theta \rightarrow \mathbb{R}^m$ be the estimand of interest. For each sample indexed $n\geq 1$, the researcher observes data $X_n \in \mathcal{X}_n$ drawn from a distribution in the model $\mathcal{P}_n=\{P_{\theta,n} : \theta \in \Theta\} \subseteq \Delta(\mathcal{X}_n)$,\footnote{Assume that $\Theta$ and each $\mathcal{X}_n$ are Polish spaces endowed with their Borel $\sigma$-algebras. Endow $\mathbb{N}$ with the power set $\sigma$-algebra, and endow product spaces with their product $\sigma$-algebras.} chooses an estimate $a \in \mathbb{R}^m$, and incurs loss $\ell_n(a-\psi(\theta))$. An \emph{estimator} for the sample $n$ is a measurable function $\delta_n : \mathcal{X}_n \times [0,1] \to \mathbb{R}^m$.\footnote{To allow for estimators which randomize conditional on the sample, we let $\delta(X_n,\cdot)$ additionally depend on the realization of an auxillary independent $\text{Unif}[0,1]$ random variable.} An \emph{estimator sequence} is a sequence of estimators $\delta = (\delta_n)_{n \geq 1}$. Let $\mathcal{D}$ be the set of estimator sequences. We begin by stating three standard assumptions on the statistical primitives which collect regularity conditions on the sequence of loss functions, the sequence of models, and the estimand of interest.

\begin{assumption}\label{as:1-loss}
$\ell_n(e)=\ell(\sqrt{n}e)$ for some Borel measurable loss function $\ell: \mathbb{R}^m \rightarrow \mathbb{R}_+$ whose sublevel sets are convex and symmetric about the origin.
\end{assumption}

Assumption \ref{as:1-loss} ensures nontrivial comparisons between $\sqrt{n}$-consistent estimator sequences, which is a standard property in parametric settings. Throughout the main text, unless otherwise stated, fix a centering value $\theta_0 \in \Theta$. Let $\overset{\theta_0}{\rightsquigarrow}$ denote convergence in distribution under the sequence of laws $P_{\theta_0,n}$. Let $\overset{\theta_0,n}{\sim}$ denote distribution under $P_{\theta_0,n}$. Say that a sequence of random variables is $o_{P_{\theta_0,n}}(1)$ if it converges in probability to $0$ under the sequence of laws $P_{\theta_0,n}$.

\begin{assumption}[\citealt{van2000asymptotic} Definition 7.14]\label{as:1}
$(\mathcal{P}_n)_n$ is \emph{locally asymptotically normal (LAN)} at $\theta_0$: there exists an invertible matrix $I_{\theta_0} \in \mathbb{R}^{k\times k}$ and random vectors $\Delta_{\theta_0,n} \overset{\theta_0}{\rightsquigarrow} N(0,I_{\theta_0})$ such that, for each converging sequence $h_n \to h$,
\begin{equation}\label{eq:LAN}
    \log \frac{dP_{\theta_0+h_n/\sqrt{n},n}}{dP_{\theta_0,n}}=h'\Delta_{\theta_0,n}-\frac{1}{2}h'I_{\theta_0}h+o_{P_{\theta_0,n}}(1) \tag{LAN}
\end{equation}
\end{assumption}

To interpret Assumption \ref{as:1}, we consider the following sequence of \emph{local informational environments} at the fixed centering value $\theta_0 \in \Theta$. Suppose that the researcher conditions on the value $\theta_0$ and instead faces uncertainty about the value of the \emph{local parameter} $h\in H=\mathbb{R}^k$,\footnote{This interpretation is often motivated by the observation that ``asymptotically, [$\theta_0$] can be known with unlimited precision," whereas the ``true statistical difficulty is$\ldots$ determined by the nature of the measures'' local to $P_{\theta_0}$ (\citealt{van2000asymptotic} Section 7.3).} where for each $n\geq 1$, $h \in H$ parametrizes the $\sqrt{n}$-local value $\theta_0+h/\sqrt{n} \in \Theta$.\footnote{Here and throughout, whenever $\theta_0+h/\sqrt{n} \notin \Theta$, set it equal to an arbitrary value in $\Theta$.} We reinterpret the data $X_n \in \mathcal{X}_n$ as arising from the \emph{local experiment} $\mathcal{P}_{\theta_0,n,H}=\{P_{\theta_0+h/\sqrt{n},n}: h\in H\}$ which is informative about the unknown local parameter $h$.\footnote{Since we have fixed the centering value $\theta_0$, $P_{\theta_0+h/\sqrt{n},n}$ is fully determined up to the local parameter $h$.} Note that as $n$ increases, two forces are at play: the informativeness of the experiment $\{P_{\theta,n}: \theta \in \Theta\}$ changes, and each parameter $\theta_0+h/\sqrt{n}$ shrinks towards $\theta_0$. \eqref{eq:LAN} requires that these two forces balance each other out, such that the sequence of local experiments $\mathcal{P}_{\theta_0,n,H}$ converges in an informational sense to the \emph{limit experiment} $\mathcal{G}_{\theta_0,H}=\{N(h,I_{\theta_0}^{-1}): h\in H\}$. For readers who are less familiar with local asymptotics, Appendix \ref{app:LAN} provides a more detailed exposition of LAN.

A leading setting satisfying Assumption \ref{as:1} is i.i.d. sampling from a smooth parametric model (\citealt{van2000asymptotic} Chapter 7).\footnote{In this case, the \emph{Fisher information} of the model $\mathcal{P}_1$ at $\theta_0$ provides the matrix $I_{\theta_0}$.} With this case in mind, we henceforth refer to the index $n\geq 1$ as the \emph{sample size}, with the understanding that our results nevertheless apply to more general statistical environments which satisfy Assumption \ref{as:1}. We conclude this subsection with a smoothness condition on the estimand of interest.

\begin{assumption}\label{as:diff-estimand}
$\psi$ is differentiable at $\theta_0$ with derivative $\dot{\psi}_{\theta_0}: \mathbb{R}^k \rightarrow \mathbb{R}^m$.  
\end{assumption}

\paragraph{Local risk function sequences.} To interpret Assumption \ref{as:diff-estimand} and future results, we introduce a unifying expositional framework for the note. We consider a sequence of \emph{local decision problems} at the fixed centering value $\theta_0$ in which the researcher conditions on the value $\theta_0$ and, for each sample size $n\geq 1$, observes data $X_n$ from the local experiment $\mathcal{P}_{\theta_0,n,H}$ about the unknown local parameter $h$, chooses an estimate $a_{\text{local}} \in \mathbb{R}^m$, and incurs loss $\ell(a_{\text{local}}-\dot{\psi}_{\theta_0,n}(h))$, where $\dot{\psi}_{\theta_0,n}: H \rightarrow \mathbb{R}^m$ is the \emph{local estimand of interest} at $\theta_0$ defined as $\dot{\psi}_{\theta_0,n}(h)=\sqrt{n}\left(\psi\left(\theta_0+h/\sqrt{n}\right)-\psi(\theta_0) \right)$. Note that $\dot{\psi}_{\theta_0,n}$ is a finite-step size approximation to the derivative $\dot{\psi}_{\theta_0}$. In this sense, Assumption \ref{as:diff-estimand} ensures the existence of a \emph{limiting} local estimand of interest.

It is convenient to summarize the performance of an estimator sequence $\delta$ in the sequence of local decision problems described above. To this end, for an estimator $\delta_n$, define its \emph{local version at $\theta_0$} to be the function $\delta_{\theta_0,n}: \mathcal{X}_n \rightarrow \Delta(\mathbb{R}^m)$ defined as: $\delta_{\theta_0,n}(X_n)=\sqrt{n}(\delta_n(X_n)-\psi(\theta_0))$. We allow $\delta_{\theta_0,n}$ to depend on $\theta_0$ since the researcher has already conditioned on the value of $\theta_0$. Loosely speaking, $\delta_{\theta_0,n}$ translates $\delta_n=\psi(\theta_0)+\delta_{\theta_0,n}/\sqrt{n}$ ``into local coordinates.'' Let $P_{\theta_0,n,h}=P_{\theta_0+h/\sqrt{n},n}$. Let $E_{\theta_0,n,h}$ denote expectation under $P_{\theta_0,n,h}$,\footnote{More precisely for randomized estimators, under $P_{\theta_0,n,h} \times \mathcal{L}(\text{Unif}[0,1])$, where $\mathcal{L}(\text{Unif}[0,1])$ is the law of an auxillary $\text{Unif}[0,1]$ random variable.} and let $\overset{\theta_0,n,h}{\sim}$ denote distribution under $P_{\theta_0,n,h}$. Let $\overline{\mathbb{R}}_+=[0,+\infty]$.

\begin{definition}\label{defn:lrfs}
The \emph{local risk function sequence of $\delta$ at $\theta_0$} is the function $R_{\theta_0,\delta}: \mathbb{N} \times H \rightarrow \overline{\mathbb{R}}_+$ defined as:
\[
R_{\theta_0,\delta}(n,h)=E_{\theta_0,n,h}\big[\ell(\delta_{\theta_0,n}(X_n)-\dot{\psi}_{\theta_0,n}(h)) \big]
\]
\end{definition}

In words, $R_{\theta_0,\delta}(n,h)$ is the expected loss (\emph{risk}) of using the local estimator $\delta_{\theta_0,n}$ to estimate the local estimand $\dot{\psi}_{\theta_0,n}(h)$ when the data is drawn according to $P_{\theta_0,n,h}$.\footnote{By substituting definitions, we may write $R_{\theta_0,\delta}(n,h)=E_{\theta_0,n,h}[\ell(\sqrt{n}(\delta_n(X_n)-\psi(\theta_0+h/\sqrt{n})))]$. This version of the expression appears in standard references (\citealt{van2000asymptotic} Ch. 8).} In this sense, $R_{\theta_0,\delta}$ encodes the performance of $\delta$ in the sequence of local decision problems at $\theta_0$, for each sample size $n$ and local parameter $h$. Say that a local risk function sequence \emph{$R$ is induced at $\theta_0$ by $\delta$} if $R=R_{\theta_0,\delta}$.

\begin{example}
\label{ex:gaussian}
\normalfont We consider as a running example the problem of estimating the mean of a one-dimensional Gaussian shift model under squared error loss with i.i.d. data. In our notation, let $\Theta=\R$, $\psi(\theta)=\theta$, $\mathcal{X}_n=\mathbb{R}$ for each $n\geq 1$, $P_{\theta,n}=N(\theta,1)^{\otimes n}$, and $\ell_n(x) = nx^2$. Assumption \ref{as:1-loss} follows since squared-error loss $\ell(x)=x^2$ satisfies the stated conditions. Fix any centering value $\theta_0 \in \mathbb{R}$, and note that the Fisher information at $\theta_0$ is $I_{\theta_0}=1$. A straightforward computation yields:
\[
\log \frac{dP_{\theta_0+h_n/\sqrt{n},n}}{dP_{\theta_0,n}}(X_n)=h_n \sqrt{n}(\Bar{X}_n-\theta_0)-\frac{h_n^2}{2}
\]
where $\Bar{X}_n$ is the sample mean and $\sqrt{n}(\Bar{X}_n-\theta_0) \overset{\theta_0,n}{\sim} N(0,1)$. Since the sequence of random variables $(h_n-h)\sqrt{n}(\Bar X_n-\theta_0)+(h^2-h_n^2)/2$ is $o_{P_{\theta_0,n}}(1)$, Assumption \ref{as:1} follows. Assumption \ref{as:diff-estimand} holds with $\dot{\psi}_{\theta_0}(h)=h$. 
The sample mean induces the constant function $R_{\theta_0,\Bar{X}}=1$.
\end{example}

\paragraph{Classical LAM risk.} The literature contains two key functional forms of LAM risk, both of which evaluate estimator sequences by their local risk function sequences. We refer to the first one as \emph{classical} LAM risk since it appears in the classical LAM lower bounds developed by \cite{hajek1972local} and \cite{le1979theorem}. Let $\cI$ be the set of nonempty finite subsets of $H$.

\begin{definition}\label{def:LAM-estimator}
The \emph{classical LAM risk} of estimator sequence $\delta$ at $\theta_0$ is:
\begin{equation}\label{eq:LAM-risk}
    \CVLAM(\delta) := \sup_{I \in \cI} \liminf_{n \to \infty} \sup_{h \in I} R_{\theta_0,\delta}(n,h) \tag{CLAM}
\end{equation}
\end{definition}

In words: for each \emph{local neighborhood} $I$ and sample size $n$, $\delta_n$ is evaluated by the \emph{worst-case} risk in the $n$-th local decision problem at $\theta_0$ over local parameters $h\in I$, the resulting sequence of worst-case risks is aggregated by its \emph{limiting best-case} across sample size subsequences, and finally the \emph{worst-case} is taken over all local neighborhoods $I$. The \emph{local asymptotic minimax theorem} (\citealt{hajek1972local,le1979theorem}; \citealt{van2000asymptotic} Theorem~8.11) lower bounds the classical LAM risk of \emph{every} estimator sequence $\delta \in \mathcal{D}$ at $\theta_0$:
\begin{equation}
\label{eq:LAM thm}
    \CVLAM(\delta) \geq B_{\mathrm{LAM},\theta_0}:=\int \ell \, dN(0, \dot{\psi}_{\theta_0}I_{\theta_0}^{-1}\dot{\psi}_{\theta_0}') \tag{LAM bound}
\end{equation}
Note that $B_{\mathrm{LAM},\theta_0}$ (or \emph{LAM bound}) is the optimal minimax risk for estimating $\dot{\psi}_{\theta_0}(h)$ with data from the limit experiment $\mathcal{G}_{\theta_0,H}$. In this sense, \eqref{eq:LAM thm} establishes that the sequence of local decision problems at $\theta_0$ is \emph{weakly more difficult than} this limit decision problem at $\theta_0$.

\paragraph{Attainment LAM risk.} 
Standard arguments in statistics and econometrics show that a proposed estimator sequence attains the LAM bound and conclude that it is efficient \citep{chamberlain1992efficiency,hirano2003asymptotic,graham2011efficiency,adusumilli2025risk}. Many of these arguments (e.g., pg. 348 of \citealt{van2002semiparametric} and Theorem~4 of \citealt{song2014local}) use an alternative risk index.

\begin{definition}
The \emph{attainment LAM risk} of estimator sequence $\delta$ at $\theta_0$ is:
\begin{equation}
\label{eq:LAM-limsup-estimator}
    \AVLAM(\delta):=\sup_{I \in \cI} \limsup_{n \to \infty} \sup_{h \in I} R_{\theta_0,\delta}(n,h) \tag{ALAM}
\end{equation}
\end{definition}

The interpretation of attainment LAM risk is exactly analogous with that of classical LAM risk, except that each sequence of worst-case risks is aggregated by its \emph{limiting worst-case} across sample size subsequences. Note that we may also write
\[
\AVLAM(\delta)=\sup_{h\in H} \limsup_{n\to\infty} R_{\theta_0,\delta}(n,h)
\]
so attainment LAM risk is also the \emph{worst-case limiting} risk taken over all local parameters $h\in H$ and all sample size subsequences. Following the above references, say that an estimator sequence $\delta$ \emph{attains the LAM bound at $\theta_0$} if $\AVLAM(\delta)\leq B_{\mathrm{LAM},\theta_0}$.

\addtocounter{example}{-1}
\begin{example}[Continued]
\label{ex:CLAM-ALAM-bounds}
\normalfont Fix any centering value $\theta_0\in \mathbb{R}$. By our previous computations, the LAM bound is $B_{\mathrm{LAM},\theta_0}=\int x^2 \, dN(0,1) = 1$. Since the local risk function sequence $R_{\theta_0,\Bar{X}}=1$ is constant, both risk indices coincide: $\AVLAM(\Bar{X})=\CVLAM(\Bar{X})=1$. Hence, the sample mean attains the LAM bound at $\theta_0$.\footnote{More generally, under mild conditions the \emph{maximum likelihood estimator} attains the LAM bound at each $\theta_0\in \Theta$. See, e.g., \cite{hirano2020asymptotic} page 311.}
\end{example}

\paragraph{Researcher preferences.} The appearance of classical and attainment LAM risk in the literature raises a natural question: when a researcher professes to apply the LAM criterion, what risk index describes their actual preferences? Although the motivation described in Section \ref{sec:intro} of normalizing the difficulty of the estimation problem to be the same order as sampling-based uncertainty suggests a local asymptotic approach in a broad sense, it is too coarse to distinguish between these two functional forms for LAM risk. Indeed, in this section we demonstrate that conventional efficiency arguments are consistent with a \emph{much larger class} of LAM-based risk indices.

To make this claim precise, recall that standard efficiency arguments conclude that an estimator sequence $\delta^*$ is optimal whenever it attains the LAM bound. Hence, say that a risk index $V_{\theta_0}: \mathcal{D} \rightarrow \overline{\mathbb{R}}_+$ is \emph{consistent} with such arguments at $\theta_0$ if, for each $\delta^* \in \mathcal{D}$,
\begin{equation}
\label{eq:cons}
    \delta^* \text{ attains the LAM bound at } \theta_0 \implies V_{\theta_0}(\delta)\geq V_{\theta_0}(\delta^*) \quad \forall \delta\in\mathcal{D} \tag{Cons.}
\end{equation}
The set of consistent risk indices is the largest set of functional forms for risk which does not contradict the property that attaining the LAM bound is a sufficient condition for global optimality. However, \eqref{eq:cons} by itself is not an appealing desideratum, since it places no restrictions on how $V_{\theta_0}$ ranks estimator sequences which do not attain the LAM bound. This means that the revealed preference exercise embodied by \eqref{eq:cons} has little bite, allowing for, e.g., the risk index which assigns $1$ to any $\delta^*$ which attains the bound and $+\infty$ to anything else. We therefore focus on the following subset of \emph{transparently consistent} risk indices.

\begin{definition}
\label{defn:obv-cons}
A risk index $V_{\theta_0}: \mathcal{D} \rightarrow \overline{\mathbb{R}}_+$ is \emph{transparently consistent} with standard efficiency arguments at $\theta_0$ if
\begin{equation}
\label{eq:obv-cons}
    \CVLAM(\delta)\leq V_{\theta_0}(\delta)\leq \AVLAM(\delta) \quad \forall \delta \in \mathcal{D} \tag{Trans. Cons.}
\end{equation}
\end{definition}

To justify this terminology, note that if $V_{\theta_0}$ satisfies \eqref{eq:obv-cons}, then verifying consistency is immediate, since $\delta^*$ attaining the LAM bound at $\theta_0$ implies
\[
\CVLAM(\delta) \geq \AVLAM(\delta^*) \quad\forall \delta\in \mathcal{D}
\]
by \eqref{eq:LAM thm}, which when combined with \eqref{eq:obv-cons} immediately yields \eqref{eq:cons}. Let $\mathcal{V}_{\theta_0}$ be the set of transparently consistent risk indices at $\theta_0$. In sum, observing that the researcher weakly prefers an estimator sequence which attains the LAM bound to every alternative does not actually specify a functional form for the LAM criterion, since such efficiency arguments are \emph{transparently consistent} with every risk index $V \in \mathcal{V}_{\theta_0}$. Our first results explain why this ambiguity may not have been formally addressed: on the choice domain of estimator sequences whose local risk function sequences converge pointwise, every transparently consistent risk index coincides.

\begin{proposition}
\label{prop:stable-local-risks}
Fix any estimator sequence $\delta$. Suppose there exists a \emph{pointwise limiting local risk function} $r_{\theta_0,\delta}: H\rightarrow \overline{\mathbb{R}}_+$ such that
\begin{equation}
\label{eq:stable-local-risks}
R_{\theta_0,\delta}(n,h)
\to r_{\theta_0,\delta}(h)
\qquad\text{for each }h\in H 
\end{equation}
Then $V_{\theta_0}(\delta)=V_{\theta_0}'(\delta)$ for all $V_{\theta_0},V_{\theta_0}'\in \mathcal{V}_{\theta_0}$.
\end{proposition}

Two transparently consistent risk indices can differ only if there exists some local parameter where local risk fails to converge across sample sizes. Our next result shows that estimator sequences whose local estimation errors possess limit laws and whose local losses are asymptotically uniformly integrable satisfy the pointwise convergence of local risk functions in \eqref{eq:stable-local-risks}. Define the \emph{local estimation error} of estimator sequence $\delta$ at sample size $n\geq 1$ and local parameter $h\in H$ to be the random vector:
\[
Z_{\theta_0,\delta,n}(h):=\delta_{\theta_0,n}(X_n)-\dot{\psi}_{\theta_0,n}(h)
\]
Recall that this is the random estimation error which arises from using the local version of $\delta_n$ at $\theta_0$ to estimate the local estimand at $\theta_0$.

\begin{corollary}
\label{cor:regular-stable-losses}
Fix any estimator sequence $\delta$. Suppose that, for each $h\in H$, there exists a random vector $Z_{\theta_0,\delta}(h)$ such that $Z_{\theta_0,\delta,n}(h) \overset{\theta_0,n,h}{\rightsquigarrow} Z_{\theta_0,\delta}(h)$, $\ell$ is a.s. continuous under the law of $Z_{\theta_0,\delta}(h)$, and the sequence of local losses $\{\ell(Z_{\theta_0,\delta,n}(h))\}_{n\geq 1}$ under the sequence of laws $P_{\theta_0,n,h}$ is \emph{asymptotically uniformly integrable}:
\begin{equation}
\label{eq:local-loss-AUI}
\lim_{M\to\infty}\limsup_{n\to\infty}
E_{\theta_0,n,h}\!\big[
\ell\!\left(Z_{\theta_0,\delta,n}(h)\right)
\mathbf 1\!\left\{
\ell\!\left(Z_{\theta_0,\delta,n}(h)\right)>M
\right\}
\big]
=0
\end{equation}
Then $V_{\theta_0}(\delta)=V_{\theta_0}'(\delta)$ for all $V_{\theta_0},V_{\theta_0}'\in \mathcal{V}_{\theta_0}$.
\end{corollary}

Definition \ref{defn:obv-cons}, Proposition~\ref{prop:stable-local-risks} and Corollary~\ref{cor:regular-stable-losses} make precise the issue at the heart of this note: conventional efficiency arguments in econometrics (and more generally, researcher rankings on a class of procedures satisfying standard regularity conditions) do not distinguish among a class of risk indices bounded by classical and attainment LAM risk. This does not imply, however, that the distinction is not meaningful. While the regularity conditions in Corollary~\ref{cor:regular-stable-losses} are standard, simple examples violate them and make the choice of $V_{\theta_0}\in \mathcal{V}_{\theta_0}$ consequential. Furthermore, a stated appeal of the LAM criterion (as opposed to, e.g., convolution arguments which restrict attention to regular estimator sequences\footnote{See, e.g., \cite{van2000asymptotic} Theorem 8.8.}) is that it provides a \emph{universal} theory of statistical optimality which does not rely on restrictions to special classes of estimators, such as the ones studied in Proposition~\ref{prop:stable-local-risks} and Corollary~\ref{cor:regular-stable-losses}. It is therefore important for the LAM criterion to evaluate such simple examples. Indeed, the following example shows that the distinction can be maximally stark.

\addtocounter{example}{-1}

\begin{example}[Continued]
\label{ex:powers-of-10}
\normalfont For ease of exposition, fix the centering value $\theta_0=0$. For any arbitrarily large but finite constant $C>0$, consider the estimator sequence 
\[
\delta_n(X_n)=\begin{cases} 
      \Bar{X}_n & n=10^k \text{ for some } k\in \mathbb{N} \\
      \Bar{X}_n+C/\sqrt{n} & \text{ else} 
\end{cases}
\]
which outputs the sample mean when the sample size is a power of 10 and otherwise outputs the sample mean with constant local bias $C$. By previous computations, 
\[
R_{\theta_0,\delta}(n,h)=\begin{cases} 
      1 & n=10^k \text{ for some } k\in \mathbb{N} \\
      1+C^2 & \text{ else} 
\end{cases}
\]
Since $\delta$ achieves the same performance as the sample mean along the subsequence $n_k=10^k$ and incurs risk $1+C^2$ otherwise,
\[
1+C^2=\AVLAM(\delta)>\CVLAM(\delta)=B_{\mathrm{LAM},\theta_0}=1
\]
which implies that $\delta$ is a \emph{best} estimator sequence under classical LAM risk for any value of $C>0$, but an arbitrarily bad estimator sequence under attainment LAM risk for sufficiently large values of $C$. More generally, a transparently consistent risk index may assign $\delta$ any risk in $[1,1+C^2]$.
\end{example}

\paragraph{Generalized $\boldsymbol{\alpha}$-LAM risk.} We conclude this section with a useful reparametrization of the set of transparently consistent risk indices, in terms of a set of \emph{estimator sequence-dependent weights} on classical and attainment LAM risk.

\begin{definition}
\label{defn-gen-alpha-LAM-estimators}
A risk index $V_{\theta_0}: \mathcal{D} \rightarrow \overline{\mathbb{R}}_+$ is \emph{generalized $\alpha$-LAM} at $\theta_0$ if there exists a function $\alpha_{\theta_0}: \mathcal{D} \rightarrow [0,1]$ such that\footnote{We use the convention $0\cdot(+\infty)=0$, such that $\alpha_{\theta_0}=1$ selects classical LAM risk and $\alpha_{\theta_0}=0$ selects
attainment LAM risk, even when an endpoint is infinite.} 
\begin{equation}
    V_{\theta_0}(\delta)=V_{\mathrm{gen}\text{-}\alpha\text{-}\mathrm{LAM},\theta_0}(\delta):=\alpha_{\theta_0}(\delta) \CVLAM(\delta)+(1-\alpha_{\theta_0}(\delta))\AVLAM(\delta) \tag{gen $\alpha$-LAM}
\end{equation}
\end{definition}

The following observation verifies that Definitions \ref{defn:obv-cons} and \ref{defn-gen-alpha-LAM-estimators} are essentially equivalent.

\begin{observation}
\label{obs:gen-alpha-LAM-iff-obv-cons}
\begin{itemize}
    \item[(i)] Every generalized $\alpha$-LAM risk index at $\theta_0$ is transparently consistent at $\theta_0$. 
    \item[(ii)] For every risk index $V_{\theta_0}$ which is transparently consistent at $\theta_0$, there exists $\alpha_{\theta_0}: \mathcal{D} \rightarrow [0,1]$ such that: 
    \[
    V_{\theta_0}(\delta)=V_{\mathrm{gen}\text{-}\alpha\text{-}\mathrm{LAM},\theta_0}(\delta) \quad \forall \delta \text{ s.t. } \CVLAM(\delta)=+\infty \text{ or } \AVLAM(\delta)<+\infty
    \]
\end{itemize}
\end{observation}

Hence, excluding estimator sequences which have finite classical LAM risk but infinite attainment LAM risk, a risk index is transparently consistent if and only if it is generalized $\alpha$-LAM. In particular, the equivalence holds over the set of estimator sequences which induce bounded local risk function sequences at $\theta_0$. Restating transparent consistency in the equivalent language of weighted averages of classical and attainment LAM risk is useful for our forthcoming decision theoretic-approach.

\section{Axiomatic results}
\label{sec:axioms}

\subsection{Roadmap of Section \ref{sec:axioms}}
In the context of Example~\ref{ex:powers-of-10}, it seems intuitive to consider $\delta$ unappealing, since it behaves well only when the sample size happens to be a power of 10. Consequently, it seems like classical LAM risk may be a poor fit for researcher preferences. To make this intuition precise, and to comprehensively understand which transparently consistent risk indices (if any) are a better fit, we take a decision-theoretic approach.

With the equivalence from Observation \ref{obs:gen-alpha-LAM-iff-obv-cons} in mind, we begin by seeking an axiomatic characterization of the set of researcher preferences with a generalized $\alpha$-LAM risk representation (Theorem \ref{thm-main-result}). We view this exercise as an initial benchmark for understanding which conditions on the researcher's rankings are necessary and sufficient to \emph{merely} be (transparently) consistent with standard efficiency arguments. Given the flexibility of the procedure-dependent weights $\alpha$, we find that this is not a particularly demanding restriction (although it does still carry behavioral content). In particular, we document that generalized $\alpha$-LAM preferences may violate basic rationality requirements, such as \emph{Monotonicity} and \emph{Mixture Continuity}. This suggests that requiring that $\alpha$ be a \emph{fixed} weight may be a more appealing fit, a class of risk representations we denote as $\alpha$-LAM.

Therefore, we next seek an axiomatic characterization of the set of $\alpha$-LAM preferences (Theorem \ref{thm:fixed-alpha-LAM}). Our axioms, which now include the basic rationality requirements mentioned above, determine how the researcher aggregates risk across sample sizes and local parameters. Within this class, two axioms specifying contrasting attitudes towards aggregating risk across sample sizes uniquely select classical LAM and attainment LAM. We argue that the latter axiom, which requires a \emph{preference for hedging}, is normatively appealing. We therefore recommend attainment LAM as the functional form for the LAM criterion. 

\subsection{Shared axioms}
We begin by stating and interpreting several variants of axioms which fulfill a shared purpose for each of the axiomatic characterizations we pursue. The first axiom clarifies our focus on optimality criteria which evaluate estimator sequences by their performance \emph{local to $\theta_0$}. For ease of exposition, we state this axiom informally in the main text and defer the formal statement to Axiom \ref{ax:cond-risk-relevance} in Appendix \ref{app:util-act-lemmas}, where we additionally take as primitive the researcher's preference over estimator sequences conditional on $\theta_0$.

\begin{axiom}[Local Risk Sufficiency, informally stated]\label{as:2}
Each estimator sequence $\delta$ is evaluated by its induced local risk function sequence $R_{\theta_0,\delta}$.
\end{axiom}

We view (the formal statement of) Axiom \ref{as:2} as a precise statement of the heuristic motivation for local asymptotics given in Section~\ref{sec:intro}. Recall that $R_{\theta_0,\delta}$ encodes $\delta$'s performance in the sequence of local decision problems at $\theta_0$ defined in Section~\ref{sec:setup}, which localize the parameter of interest around $\theta_0$ at the same rate as sampling uncertainty vanishes. Axiom \ref{as:2} therefore ensures that performance in this sequence of decision problems is \emph{all that matters} for ranking estimator sequences. Recall that classical and attainment LAM risk both satisfy Axiom \ref{as:2}. More generally, this holds for every generalized $\alpha$-LAM risk index where $\alpha_{\theta_0}(\delta)$ depends on $\delta$ only through $R_{\theta_0,\delta}$. Given that such heuristic motivations for local asymptotics are standard, we view Axiom \ref{as:2} as a natural restriction for studying local asymptotic optimality.\footnote{At this stage, the distinction between the local risk function sequence induced by $\delta$ at $\theta_0$ and the risk function sequence induced by $\delta$ at $\theta_0$ is vacuous: one may be obtained from the other by the usual reparametrization defined in Section \ref{sec:setup}. In this sense, Local Risk Sufficiency is equivalent to Risk Sufficiency. However, the distinction has bite when paired with our forthcoming axioms.} However, it rules out criteria which evaluate estimator sequences by other desiderata beyond their performance local to $\theta_0$, such as simplicity or computational complexity.

\addtocounter{example}{-1}
\begin{example}[Continued]
\normalfont Axiom \ref{as:2} requires that $(\Bar{X}_n)_n$ is preferred to $\delta$ if and only if $R_{\theta_0,\Bar{X}}$ is preferred to $R_{\theta_0,\delta}$.
\end{example}

\paragraph{Preferences over local risk function sequences.} With Axiom \ref{as:2} in hand, it is without loss of generality to directly study the researcher's preferences at $\theta_0$ over \emph{local risk function sequences}. A \emph{local risk function} is a bounded, measurable function $r: H \rightarrow \mathbb{R}_+$. Let $R_H$ be the set of local risk functions. A \emph{local risk function sequence} is a bounded, measurable function $R: \mathbb{N} \times H \rightarrow \mathbb{R}_+$. Let $\mathcal{R}_H$ be the set of local risk function sequences. We take as primitive a binary relation $\succsim_{\theta_0}$ on $\mathcal{R}_H$, which we interpret as the preferences of a researcher facing the sequence of local informational environments defined in Section \ref{sec:setup}, in which the researcher has conditioned on the fixed centering value $\theta_0$. 

Recall from Section \ref{sec:intro} that, although in practice the researcher’s problem features a fixed sample size, the criteria we study operate under a convenient fiction in which the researcher ranks local risk function \emph{sequences}
indexed by sample size. We may therefore interpret the decision environment as one in
which $\succsim_{\theta_0}$ is formed \emph{before} the researcher learns the sample size of their dataset. Under this interpretation, at the time of their decision, the researcher faces \emph{uncertainty} about the sample size $n$ (in addition to the local parameter $h$),\footnote{Since $\succsim_{\theta_0}$ is conditional on $\theta_0$, the representation results developed in the main text are also stated conditional on $\theta_0$. Appendix \ref{app:util-act-lemmas} studies a more primitive setup which does not condition on a fixed $\theta_0$, and adds an \emph{unanimity} axiom across centering values (Axiom \ref{ax:Theta-unanimity}) to ensure that the criterion considers performance conditional on $\theta_0$ for each $\theta_0\in \Theta$.} and their stance towards how to trade off risk across $n$ and $h$ is determined by their \emph{attitudes towards uncertainty} about $n$ and $h$. Following standard conventions in the microeconomic decision theory literature, we draw on this uncertainty-based interpretation to organize and discuss our axioms.

We briefly discuss the decision environment. While the boundedness restriction on the choice domain $\mathcal{R}_H$ is standard in the literature on decision-making under uncertainty, it rules out some local risk function sequences induced by estimator sequences where risk tends to infinity as the sample gets large (or yields infinite risk in finite sample). The LAM risk representations we obtain from our axioms therefore do not apply to such local risk function sequences. As is standard in the literature studying optimality criteria from a decision-theoretic perspective \citep{stoye2012new,andrews2026misspecification}, we study the researcher's preferences over \emph{all} local risk function sequences, even though the set of (bounded and measurable) local risk function sequences induced by some estimator sequence is a strict subset. To illustrate our axioms, we sometimes ask the researcher to consider their preferences over infeasible local risk function sequences, including those induced by \emph{oracle estimator sequences} which may condition their estimates on $h$. We have endeavored to make the sequences we draw on simple and interpretable, regardless of their feasibility.

We introduce some useful notation. For each $k\geq 0$, let $\overrightarrow{k} \in R_H$ denote the \emph{constant local risk function} which yields risk $k$ for each $h \in H$.\footnote{When the context is clear, we abuse notation and also use $\overrightarrow{k}$ to refer to the \emph{constant sequence} of constant local risk functions $(\overrightarrow{k},\overrightarrow{k},\ldots) \in \mathcal{R}_H$.} Estimators which are \emph{equivariant-in-law}\footnote{An estimator $\delta_n$ is \emph{equivariant-in-law} if the law of $\sqrt{n}(\delta_n(X_n)-\psi(\theta))$ under $X_n \sim P_{\theta,n}$ does not depend on $\theta$.} (such as the sample mean in Example \ref{ex:gaussian}) yield constant local risk functions. For local risk function sequences $R,R' \in \mathcal{R}_H$ and $\alpha \in [0,1]$, define the mixture $\alpha R+(1-\alpha)R' \in \mathcal{R}_H$ as: $(\alpha R+(1-\alpha)R')(n,h)=\alpha R(n,h)+(1-\alpha)R'(n,h)$. Note that this mixture occurs pointwise in $(n,h)$. To interpret this object, suppose $R$ is induced at $\theta_0$ by $\delta$ and $R'$ is induced at $\theta_0$ by $\delta'$. Then, $\alpha R+(1-\alpha)R'$ is induced at $\theta_0$ by the estimator sequence which flips a coin (independently of the data) whose probability of heads is $\alpha$, and uses $\delta$ if heads and $\delta'$ if tails.\footnote{Since the local risk function sequence induced by an estimator sequence only depends on the set of \emph{marginal distributions} of estimation errors for each $(n,h)$, it does not matter whether the coin is flipped before or after $(n,h)$ is realized (as long as the true value of $(n,h)$ does not depend on the outcome of the coin flip).}

We emphasize that comparing two local risk function sequences $R,R'\in \mathcal{R}_H$ is a nontrivial task, since $R$ may yield lower risk than $R'$ at some values of $(n,h)$ and higher risk at other values. Any complete comparison must therefore take a stance on how to trade off risk across sample sizes $n\geq 1$ and local parameters $h\in H$. A \emph{risk index} $V: \mathcal{R}_H \rightarrow \mathbb{R}$ provides a complete ranking of local risk function sequences in $\mathcal{R}_H$, where $V(R)$ aggregates the performance of $R$ across each $(n,h)$ into a single index of risk. A risk index $V: \mathcal{R}_H \rightarrow \mathbb{R}$ is a \emph{risk representation of $\succsim_{\theta_0}$ on $\mathcal{R}_H$} if: for each $R,R' \in \mathcal{R}_H$, $R\succsim_{\theta_0} R'$ if and only if $V(R)\leq V(R')$. The following definition collects the key LAM-based risk representations of interest.

\begin{definition}\label{defn:LAM-risk-main-text}
\begin{itemize}
    \item[(i)] $\succsim_{\theta_0}$ has a \emph{generalized $\alpha$-LAM risk representation on $\mathcal R_H$} if there exists a function $\alpha: \mathcal{R}_H \rightarrow [0,1]$ such that the function
    \begin{equation}\label{eq:gen-LAM-risk-act}
        V_{\mathrm{gen}\text{-}\alpha\text{-}\mathrm{LAM}}(R)
        :=\alpha(R)\underline V_{\mathrm{LAM}}(R)+(1-\alpha(R)) \overline V_{\mathrm{LAM}}(R) \tag{gen-$\alpha$-LAM risk}
    \end{equation}
    represents it, where
    \begin{equation}
    \label{eq:clam-risk}
        \underline V_{\mathrm{LAM}}(R):=\sup_{I\in\mathcal I}\liminf_{n\to\infty}\max_{h\in I}R(n,h) \tag{CLAM risk}
    \end{equation}
    and
    \begin{equation}
    \label{eq:alam-risk}
        \overline V_{\mathrm{LAM}}(R):=\sup_{I\in\mathcal I}\limsup_{n\to\infty}\max_{h\in I}R(n,h) \tag{ALAM risk}
    \end{equation}
    \item[(ii)] $\succsim_{\theta_0}$ has an \emph{$\alpha$-LAM risk representation on $\mathcal R_H$} if it has a generalized $\alpha$-LAM risk representation on $\mathcal R_H$ where $\alpha \in [0,1]$ is a constant.
\end{itemize}
\end{definition}

Recall that the set of generalized $\alpha$-LAM researcher preferences on $\mathcal{R}_H$ is precisely the set of researcher preferences on $\mathcal{R}_H$ that have a risk representation that is \emph{transparently consistent} with standard efficiency arguments.\footnote{Since local risk function sequences in $\mathcal{R}_H$ are bounded, the equivalence holds exactly: $V$ is generalized $\alpha$-LAM on $\mathcal{R}_H$ if and only if $\underline{V}_{\mathrm{LAM}}\leq V\leq \overline{V}_{\mathrm{LAM}}$ on $\mathcal{R}_H$.} As an initial benchmark, we axiomatically characterize this set. Our axioms reveal a lack of basic rationality requirements such as Monotonicity and Mixture Continuity. Amending our axioms to account for these will lead us to the (fixed) $\alpha$-LAM class, and eventually to the classical and attainment LAM risk representations \eqref{eq:clam-risk} and \eqref{eq:alam-risk}.

Our axiomatic exercise is tightly connected to the motivation and decision environment of Section \ref{sec:setup}. In particular, each risk index $V: \mathcal{R}_H \rightarrow \mathbb{R}_+$ studied in Definition \ref{defn:LAM-risk-main-text}(i) and (ii) induces a corresponding optimality criterion for estimator sequences conditional on $\theta_0$, by defining the risk index of $\delta$ as $V_{\theta_0}(\delta):=V(R_{\theta_0,\delta})$. We therefore seek axiomatic characterizations of the set of binary relations $\succsim_{\theta_0}$ on $\mathcal{R}_H$ which have the risk representations in Definition \ref{defn:LAM-risk-main-text}(i) and (ii) with the understanding that, under Axiom~\ref{as:2}, obtaining such a representation for local risk function sequences produces the desired representation for estimator sequences.\footnote{More precisely, for estimator sequences whose induced local risk function sequences at $\theta_0$ are bounded, measurable
functions $\mathbb{N} \times H\rightarrow \mathbb{R}_+$.}

\paragraph{Basic axioms.} We begin with two variants of standard regularity conditions on preferences. We use the conditions in Axiom \ref{ax: basics-main-text} as part of our characterization of $\alpha$-LAM preferences. However, in this section we document that generalized $\alpha$-LAM need not satisfy even these standard conditions, including Monotonicity and Mixture Continuity. Consequently, Axiom \ref{ax:gen-basics-main-text} provides a weaker variant to use for our characterization of generalized $\alpha$-LAM.

\begin{axiom}[Basic axioms]\label{ax: basics-main-text}
\begin{itemize}
    \item[(i)] \textnormal{Weak Order:} $\succsim_{\theta_0}$ is complete and transitive on $\mathcal{R}_H$. 
    \item[(ii)] \textnormal{Strong Monotonicity:} For each $R,R' \in \mathcal{R}_H$,
    \[
    R'(n,h)\geq R(n,h) \quad \forall (n,h)\in \mathbb{N} \times H \implies R \succsim_{\theta_0} R'
    \]
    For any $j>k$, $\overrightarrow{k} \succ_{\theta_0} \overrightarrow{j}$.
    \item[(iii)] \textnormal{Mixture Continuity:} For each $R,R',R'' \in \mathcal{R}_H$, the following sets are closed:
    \[
    \{\alpha \in [0,1]: \alpha R+(1-\alpha)R' \succsim_{\theta_0} R''\} \quad \text{and} \quad \{\alpha \in [0,1]: R'' \succsim_{\theta_0} \alpha R+(1-\alpha)R'\}
    \]
\end{itemize}
\end{axiom}

Axiom \ref{ax: basics-main-text}(i) requires that the researcher can rank any two local risk function sequences, and that their rankings are transitive.\footnote{Although completeness is a standard axiom in the microeconomic decision theory literature, it may reasonably be viewed as strong in the context of comparing local risk function sequences. Nevertheless, completeness is required by any risk representation.} Axiom \ref{ax: basics-main-text}(ii) requires that a local risk function sequence which yields weakly lower risk at every sample $n$ and local parameter value $h$ is weakly preferred, and strictly lower values of constant risk are strictly preferred. Axiom \ref{ax: basics-main-text}(iii) requires that small perturbations in mixtures of local risk function sequences do not alter their ranking.

The following example documents that, due to the arbitrary dependence of $\alpha(R)$ on $R$, generalized $\alpha$-LAM preferences may violate Monotonicity and Mixture Continuity.

\addtocounter{example}{-1}
\begin{example}[Continued]
\label{ex:violations}
\normalfont Consider the statistical environment of our running example. For a bounded sequence $(a_n)_n$, consider the estimator sequence $\gamma_n=\Bar{X}_n+a_n/\sqrt{n}$, which shifts the sample mean by local bias $a_n$. Since each $\gamma_n$ is equivariant-in-law, analogous computations as before yield that $R_{\theta_0,\gamma}=(\overrightarrow{1+a_n^2})_{n}$. This implies that every bounded sequence of constant local risk functions $K=(\overrightarrow{k_n})_n$ with each $k_n\geq 1$ is induced by some $\gamma$ and $(a_n)_n$. We therefore work directly with such local risk function sequences. For ease of notation, we drop the overhead arrow and refer to $K=(k_n)_n$ as a \emph{local risk sequence}.
\begin{itemize}
    \item[(i)] Consider the local risk sequences $K=(1,5,1,5,\ldots)$ and $K'=(2,10,2,10,\ldots)$, respectively. Set $\alpha(K)=0$ and $\alpha(K')=1$, and extend $\alpha$ arbitrarily to $\mathcal{R}_H$. Then $K'\geq K$ pointwise, but $V_{\mathrm{gen}\text{-}\alpha\text{-}\mathrm{LAM}}(K)=5>2=V_{\mathrm{gen}\text{-}\alpha\text{-}\mathrm{LAM}}(K')$, which implies $K'\succ_{\theta_0}K$: the sequence with strictly higher risk at every sample size is strictly preferred. This violates Monotonicity.
    \item[(ii)] For each $\beta \in [0,1]$, define $R_\beta:=\beta K+(1-\beta)K'$. Set $\widetilde\alpha(R_\beta)=1$ for rational $\beta$ and $0$  otherwise, and extend $\widetilde\alpha$ to $\mathcal{R}_H$ arbitrarily. Then
    \[
    \{\beta \in [0,1]:R_\beta\succsim_{\theta_0}\overrightarrow 3\}=\mathbb{Q} \cap [0,1]
    \]
    which is not closed. This violates Mixture Continuity.
\end{itemize}
\end{example}

We must therefore weaken Axiom \ref{ax: basics-main-text} to characterize generalized $\alpha$-LAM.

\begin{axiom}[Basic Axioms II]\label{ax:gen-basics-main-text}
\begin{itemize}
    \item[(i)] \textnormal{Weak Order:} $\succsim_{\theta_0}$ is complete and transitive on $\mathcal R_H$.
    \item[(ii)] \textnormal{Constant Calibration:} For each $k,j\in\mathbb R_+$, $\overrightarrow{k}\succsim_{\theta_0}\overrightarrow{j}
    \iff k\leq j$.
    \item[(iii)] \textnormal{Constant Mixture Continuity:} For each $R\in\mathcal R_H$ and $k,j\in\mathbb R_+$, the following sets are closed:
    \[
    \left\{\beta \in [0,1]: \beta \overrightarrow{k}+(1-\beta)\overrightarrow{j} \succsim_{\theta_0}R\right\}
    \quad \text{and} \quad
    \left\{\beta \in [0,1]:R\succsim_{\theta_0}\beta \overrightarrow{k}+(1-\beta)\overrightarrow{j} \right\}
    \]
\end{itemize}
\end{axiom}

Axiom \ref{ax:gen-basics-main-text}(i) is identical to Axiom \ref{ax: basics-main-text}(i). Axiom \ref{ax:gen-basics-main-text}(ii) merely requires that (strictly) smaller risks are (strictly) better. Axiom \ref{ax:gen-basics-main-text}(iii) requires that small perturbations in mixtures of \emph{constant risks} do not alter their ranking. Note that each of Axiom \ref{ax:gen-basics-main-text}(i)-(iii) is implied by Axiom \ref{ax: basics-main-text}(i)-(iii), respectively.

\addtocounter{example}{-1}
\begin{example}[Continued]
\label{ex:gamma}
\normalfont The estimator sequence $\gamma_n=\Bar{X}_n+1/\sqrt{n}$ induces the constant local risk function sequence $R_{\theta_0,\gamma}=\overrightarrow{2}$. Both Axioms \ref{ax: basics-main-text}(ii) and \ref{ax:gen-basics-main-text}(ii) require that $R_{\theta_0,\Bar{X}} \succ_{\theta_0} R_{\theta_0,\gamma}$. 
\end{example}

\paragraph{Axioms for uncertainty towards $\boldsymbol{I}$.} Observe that each of the risk representations we consider may be written as $V(R)=\sup_{I\in \mathcal{I}} V_I(R)$ for a collection of risk indices $\{V_I\}_{I\in \mathcal{I}}$, which takes the \emph{worst-case risk} over all local neighborhoods $I\in \mathcal{I}$. In this section, we state axioms which ensure that the researcher's attitudes towards uncertainty about $I$ are described by this worst-case evaluation. Representations related to this form are axiomatically studied in \cite{gilboa2010objective}, and our axioms are partially adapted from their analysis.\footnote{\cite{gilboa2010objective} study the \emph{maxmin expected utility} model of \cite{gilboa1989maxmin}, which takes the worst case over \emph{beliefs}. In their case, preferences conditional on a belief are \emph{subjective expected utility}. By contrast, in our setting the researcher takes the worst case over \emph{local neighborhoods} $I\in \mathcal{I}$, and conditional on $I$, preferences have representations which require their own axiomatic analysis. \cite{gilboa2010objective} also make the additional assumption that the unanimity rule induced by conditional preferences is separately observed, whereas we derive the $I$-conditional preferences directly from the researcher's actual preference. For these reasons, our approach requires different techniques (such as Lemma \ref{lem:I-rep} in Appendix \ref{app:util-act-lemmas}).} However, since we study generalized $\alpha$-LAM risk and $\alpha$-LAM risk, we again require two variants of these axioms specialized to our setting and functional forms of interest.

To define these axioms, we first derive the researcher's preferences over $\mathcal{R}_H$ \emph{conditional} on each local neighborhood $I \in \mathcal{I}$, denoted $\succsim_{\theta_0,I}$, from the researcher's actual preferences $\succsim_{\theta_0}$. Our axioms are then jointly imposed on $\succsim_{\theta_0}$ and $\{\succsim_{\theta_0,I}\}_{I\in\mathcal{I}}$. We begin by defining a useful \emph{splicing} operation on the set of local risk function sequences.

\begin{definition}
For each $R,R' \in \mathcal{R}_H$ and $I \in \mathcal{I}$, define $R_I R' \in \mathcal{R}_H$ as:
\[
(R_IR')(n,h)=\begin{cases}
    R(n,h) & \text{if } h\in I \\
    R'(n,h) & \text{else}
\end{cases}
\]
\end{definition}

To interpret the splicing operation, suppose $R$ and $R'$ are induced at $\theta_0$ by $\delta$ and $\delta'$. Then $R_I R'$ is induced by an oracle estimator sequence which uses $\delta$ when $h\in I$ and $\delta'$ otherwise. $I$-conditional preferences are derived by comparing spliced local risk function sequences after fixing $R'$ to be the best constant risk $\overrightarrow{0}$.

\begin{definition}\label{def:I-cond}
For each $I\in\mathcal I$, define the binary relation $\succsim_{\theta_0,I}$ on $\mathcal R_H$ as:
\[
R\succsim_{\theta_0,I}R'
\iff
R_I\overrightarrow 0\succsim_{\theta_0}R'_I\overrightarrow 0
\]
\end{definition}

We refer to $\succsim_{\theta_0,I}$ as the researcher's \emph{$I$-conditional preference}. To interpret it, let $R$ and $R'$ be induced at $\theta_0$ by $\delta$ and $\delta'$, and suppose that for each $h\notin I$, the oracle perfectly reveals the local estimand at $h$ for every $n\geq 1$, such that the researcher incurs \emph{zero} risk. Then, the researcher $I$-conditionally prefers $R$ to $R'$ when they nevertheless prefer to use $\delta$ over $\delta'$ for local parameters $h\in I$ where the local estimand is not revealed.

As discussed, each of our axioms in this section relate the preferences $\succsim_{\theta_0}$ and $\{\succsim_{\theta_0,I}\}_{I\in \mathcal{I}}$ to specify a \emph{maximally pessimistic} attitude towards uncertainty about $I\in \mathcal{I}$.

\begin{axiom}[Consistency and Caution]\label{ax:cons-and-caution-main-text}
\begin{itemize}
    \item[(i)] \emph{Consistency:} For each $R,R' \in \mathcal{R}_H$, 
    \[
    R \succsim_{\theta_0,I} R' \quad \forall I \in \mathcal{I} \implies R \succsim_{\theta_0} R'
    \]
    \item[(ii)] \emph{Caution:} For each $R \in \mathcal{R}_H$ and $k\in \mathbb{R}_+$,
    \[
    \overrightarrow{k}\succ_{\theta_0,I} R \quad \text{ for some } I\in\mathcal{I} \implies \overrightarrow{k}\succsim_{\theta_0} R
    \]
\end{itemize}
\end{axiom}

Axiom \ref{ax:cons-and-caution-main-text}(i) requires that if a local risk function sequence is $I$-conditionally preferred \emph{for each local neighborhood $I\in \mathcal{I}$}, then it is unconditionally preferred. Since Axiom \ref{ax:cons-and-caution-main-text}(i) only requires that $\succsim_{\theta_0}$ respects $\{\succsim_{\theta_0,I}\}_{I\in\mathcal{I}}$ \emph{when the latter unanimously agree} (and is silent otherwise), it ensures that $\succsim_{\theta_0}$ is \emph{minimally consistent} with $\{\succsim_{\theta_0,I}\}_{I\in\mathcal{I}}$.\footnote{In this sense, we may also interpret Axiom \ref{ax:cons-and-caution-main-text}(i) as a Pareto principle for aggregating $\{\succsim_{\theta_0,I}\}_{I\in \mathcal{I}}$.} Axiom \ref{ax:cons-and-caution-main-text}(ii) ensures that as long as there exists \emph{some} local neighborhood $I\in \mathcal{I}$ for which the researcher would $I$-conditionally prefer to incur a certain risk over a local risk function sequence, they would unconditionally prefer to do so as well. Put another way, say that a preference exhibits \emph{caution} if it ranks a certain risk over a local risk function sequence (whose risk may depend on the uncertain values of $n$ and $h$). Axiom \ref{ax:cons-and-caution-main-text}(ii) requires that as soon as some $\succsim_{\theta_0,I}$ (strictly) exhibits caution, $\succsim_{\theta_0}$ does so (weakly) as well. In this sense, $\succsim_{\theta_0}$ is \emph{maximally cautious} with respect to $\{\succsim_{\theta_0,I}\}_{I\in\mathcal{I}}$.

We use the conditions in Axiom \ref{ax:cons-and-caution-main-text} as part of our characterization of $\alpha$-LAM. However, as in the previous section, unrestricted dependence of $\alpha(R)$ on $R$ means that generalized $\alpha$-LAM need not satisfy Consistency and Caution.

\addtocounter{example}{-1}
\begin{example}[Continued]
\normalfont Consider the local risk sequences $K=(1,10,1,10,\ldots)$ and $K'=(2,5,2,5,\ldots)$. As previously discussed, each of these local risk sequences is feasible in the statistical environment of our running example. Suppose $\succsim_{\theta_0}$ has a generalized $\alpha$-LAM representation $V_{\mathrm{gen}\text{-}\alpha\text{-}\mathrm{LAM}}$.
\begin{itemize}
    \item[(i)] For each $I\in \mathcal{I}$, set $\alpha(K_I\overrightarrow{0})=\alpha(K'_I\overrightarrow{0})=1$ and $\alpha(K)=\alpha(K')=0$. Then,
    \begin{align*}
        V_{\mathrm{gen}\text{-}\alpha\text{-}\mathrm{LAM}}(K'_I\overrightarrow{0})\geq V_{\mathrm{gen}\text{-}\alpha\text{-}\mathrm{LAM}}(K_I\overrightarrow{0}) \quad \forall I\in \mathcal{I} \\
        \iff 2=\liminf_{n\to\infty} K'_n\geq \liminf_{n\to\infty} K_n=1
    \end{align*}
    but
    \begin{align*}
        10=\limsup_{n\to\infty} K_n=V_{\mathrm{gen}\text{-}\alpha\text{-}\mathrm{LAM}}(K)>V_{\mathrm{gen}\text{-}\alpha\text{-}\mathrm{LAM}}(K')=\limsup_{n\to\infty} K'_n=5
    \end{align*}
    \item[(ii)] Fix any $I\in \mathcal{I}$ and set $\alpha(K_I\overrightarrow{0})=0$, $\alpha(K)=1$, and $k=5$. Then, $k<\limsup_{n\to\infty} K_n=10$ but $k>\liminf_{n\to\infty} K_n=1$.
\end{itemize}
\end{example}

For generalized $\alpha$-LAM, we therefore instead require \emph{limiting versions} of Consistency and Caution.

\begin{axiom}[Asymptotic Consistency and Caution]
\label{ax:asymptotic-benchmark-dominance}
For each $R\in\mathcal R_H$ and $k\in\mathbb R_+$,
\begin{enumerate}[(i)]
    \item[(i)] \emph{Asymptotic Consistency:}
    \[
    R_n \succsim_{\theta_0,I} \overrightarrow{k} \text{ e.v.} \quad \forall I\in \mathcal{I} \implies R\succsim_{\theta_0}\overrightarrow{k}
    \]
    \item[(ii)] \emph{Asymptotic Caution:}
    \[
    \overrightarrow{k}\succsim_{\theta_0,I} R_n \text{ e.v.} \quad \text{ for some } I \in \mathcal{I} \implies \overrightarrow{k}\succsim_{\theta_0}R
    \]
\end{enumerate}
\end{axiom}

Axiom \ref{ax:asymptotic-benchmark-dominance}(i)-(ii) are asymptotic analogs of Axioms \ref{ax:cons-and-caution-main-text}(i)-(ii). Axiom \ref{ax:asymptotic-benchmark-dominance}(i) requires that if each local risk function along a sequence is \emph{eventually} $I$-conditionally preferred to a constant risk for \emph{every} local neighborhood $I\in \mathcal{I}$, then the entire sequence is unconditionally preferred. Axiom \ref{ax:asymptotic-benchmark-dominance}(ii) requires that as soon as each local risk function along a sequence is \emph{eventually} $I$-conditionally dominated by a certain risk for \emph{some} local neighborhood, then the entire sequence is unconditionally dominated.

Since Axioms \ref{ax:cons-and-caution-main-text} and \ref{ax:asymptotic-benchmark-dominance} characterize the aggregator $\sup_{I\in \mathcal{I}}$ in the functional forms of $\alpha$-LAM risk and generalized $\alpha$-LAM risk, our remaining axioms are applied to $\succsim_{\theta_0,I}$ for each $I\in \mathcal{I}$. We make a further simplifying observation: $R$'s ranking according to $\succsim_{\theta_0,I}$ depends on $R$ only through its restriction to $I$, denoted $R_I:\mathbb N\times I\to\mathbb R_+$. The interpretation of this observation is straightforward: if two estimator sequences yield the same risk at every $(n,h)$ with $h\in I$, an oracle who is constrained to provide zero risk outside $I$ has no scope to alter the risk profile. Hence, the researcher is $I$-conditionally indifferent between them.

Let $\mathcal{R}_I$ be the set of bounded, measurable functions $R_I: \mathbb{N} \times I \rightarrow \mathbb{R}_+$. An \emph{$I$-local risk function sequence} $R_I \in \mathcal{R}_I$ may be interpreted analogously as before: if $R$ is induced at $\theta_0$ by $\delta$, $R_I$ records $\delta$'s performance in the sequence of local decision problems at $\theta_0$ for local parameters $h\in I$. The above observation ensures that we may impose the remaining axioms on the induced $I$-conditional preference $\succsim_{\theta_0,I}$ on the simpler domain $\mathcal{R}_I$.\footnote{This follows from the axioms we have imposed on $\succsim_{\theta_0}$ and by definition of $\succsim_{\theta_0,I}$. Lemma \ref{lem:I-relevance} in Appendix \ref{app:util-act-lemmas} formalizes this observation.} 

\paragraph{Axioms for uncertainty towards $\boldsymbol{h}$.} Observe that, for a fixed local neighborhood $I\in \mathcal{I}$ and sample size $n\geq 1$, each of the risk representations we study depend on $R(n,\cdot)$ only through its \emph{worst-case risk} over $h\in I$: $\max_{h\in I} R(n,h)$. In this section, we offer axioms which ensure that the researcher's $I$-conditional attitudes towards uncertainty about $h\in I$ are described by this worst-case evaluation, and which apply to each of the risk representations we consider. Representations related to the form $R(n,\cdot) \mapsto \max_{h\in I} R(n,h)$ were axiomatized by \cite{gilboa1989maxmin} (henceforth GS89), and our axioms adapt their analysis to our setting. 

To isolate the researcher's attitudes towards uncertainty about $h\in I$, we focus on the researcher's ranking of elements of $\mathcal{R}_I$ which are \emph{constant in $n$}. Formally, an \emph{$I$-local risk function} is a function $r_I: I \rightarrow \mathbb{R}_+$. Let $\mathbb{R}_+^I \subseteq \mathcal{R}_I$ denote the set of $I$-local risk functions.\footnote{By identifying each $r_I \in \mathbb{R}_+^I$ with the $I$-local risk function sequence $(r_I,r_I,\ldots) \in \mathcal{R}_I$, we may view $\mathbb{R}_+^I$ as a subset of $\mathcal{R}_I$. Hence, any binary relation on $\mathcal{R}_I$ induces a binary relation on $\mathbb{R}_+^I$.} To interpret this object, suppose $R_I$ is induced at $\theta_0$ by $\delta$. For each fixed $n\geq 1$, $(R_I)_n$ is the $I$-local risk function induced by $\delta_n$ at $\theta_0$, which describes $\delta_n$'s performance in the $n$-th local decision problem at $\theta_0$ for local parameters $h\in H$.

\begin{axiom}[GS89 Axioms]\label{ax:gs-main-text}
$\succsim_{\theta_0,I}$ on $\mathbb{R}_+^I$ satisfies:
\begin{itemize}
    \item[(i)] \textnormal{Nontrivial Weak Order:} $\succsim_{\theta_0,I}$ is nontrivial, complete and transitive on $\mathbb{R}_+^I$.
    \item[(ii)] \textnormal{Monotonicity:} For each $r_I,r_I' \in \mathbb{R}_+^I$, $r_I'(h)\geq r_I(h)$ for all $h\in I$ implies $r_I\succsim_{\theta_0,I} r_I'$.
    \item[(iii)] \textnormal{Mixture Continuity:} For each $r_I,r_I',r_I'' \in \mathbb{R}_+^I$, the following sets are closed:
    \[
    \{\alpha \in [0,1]: \alpha r_I+(1-\alpha)r_I' \succsim_{\theta_0,I} r_I''\} \quad \text{and} \quad \{\alpha \in [0,1]: r_I'' \succsim_{\theta_0,I} \alpha r_I+(1-\alpha)r_I'\}
    \]
    \item[(iv)] \textnormal{Certainty Independence:} For each $r_I,r_I' \in \mathbb{R}_+^I$, $k \in \mathbb{R}_+$, and $\alpha \in (0,1)$, 
    \[
    r_I \succsim_{\theta_0,I} r_I' \iff 
    \alpha r_I + (1-\alpha)\overrightarrow{k} \succsim_{\theta_0,I} \alpha r_I' + (1-\alpha)\overrightarrow{k}
    \]
    \item[(v)] \textnormal{Uncertainty Aversion}: For each $r_I,r_I' \in \mathbb{R}_+^I$ and $\alpha \in (0,1)$, 
    \[
    r_I \sim_{\theta_0,I} r_I' \implies \alpha r_I + (1-\alpha)r_I' \succsim_{\theta_0,I} r_I
    \]
\end{itemize}
\end{axiom}

Axiom \ref{ax:gs-main-text}(i)-(iii) have analogous interpretations to Axiom \ref{ax: basics-main-text}(i)-(iii), adapted to rankings of $I$-local risk functions. Axiom \ref{ax:gs-main-text}(iv) requires that rankings of $I$-local risk functions are preserved under mixing with (identical) constant local risks. To understand this, suppose $r_I$ and $r_I'$ are induced by estimators $\delta_n$ and $\delta_n'$ at $\theta_0$, and let $\gamma_n$ be any equivariant-in-law estimator. Axiom \ref{ax:gs-main-text}(iv) requires that if the researcher $I$-conditionally prefers $\delta_n$ to $\delta_n'$, then they $I$-conditionally prefer the randomized estimator which flips a coin and uses $\delta_n$ if heads and $\gamma_n$ if tails over the randomized estimator flips a coin and uses $\delta_n'$ if heads and $\gamma_n$, regardless of the distribution of the coin. Axiom \ref{ax:gs-main-text}(v) ensures that the researcher exhibits a \emph{preference for hedging} against uncertainty towards $h$: if the researcher is $I$-conditionally indifferent between $\delta_n$ and $\delta_n'$, then they $I$-conditionally prefer the randomized estimator which flips a coin between them to either $\delta_n$ or $\delta_n'$, regardless of the distribution of the coin. Intuitively, the coin acts as an external randomization device which flattens the risk profile across $h\in I$.\footnote{In statistics and econometrics, we often have the intuition (motivated by Jensen’s inequality and convex loss functions) that adding extra randomness is undesirable. Axioms such as Axiom \ref{ax:gs-main-text}(v) do not contradict this intuition: they do not require randomized estimators to be \emph{undominated}, merely that some randomized estimators dominate other estimators of a particular form.}

The GS89 axioms ensure that $\succsim_{\theta_0,I}$ on $\mathbb{R}_+^I$ has a representation of the form $r_I \mapsto \max_{p\in C} \int r_I \ dp$ for some set of priors $C\subseteq \Delta(I)$. To ensure that the researcher takes the worst-case over \emph{all} $h\in I$, we impose the following axiom, adapted from \cite{ghirardato2002ambiguity}.

\begin{axiom}[$I$-Maximal Caution]\label{ax:I-Caution-Main-Text}
For each $r_I \in \mathbb R_+^I$ and $k\in \mathbb{R}_+$,
\[
r_I \succsim_{\theta_0,I} \overrightarrow{k} \implies \overrightarrow{r_I(h)} \succsim_{\theta_0,I} \overrightarrow{k}  \quad \forall h\in I
\]
\end{axiom}

Axiom \ref{ax:I-Caution-Main-Text} requires that whenever the researcher $I$-conditionally prefers an estimator $\delta_n$ to an equivariant-in-law estimator $\gamma_n$, $\delta_n$ must incur less risk than $\gamma_n$ for every $h\in I$.

\paragraph{Characterizing generalized $\boldsymbol{\alpha}$-LAM.} The following result provides an axiomatic characterization of generalized $\alpha$-LAM.

\begin{theorem}[Generalized $\alpha$-LAM]\label{thm-main-result}
The following are equivalent:
\begin{itemize}
    \item[(i)] $\succsim_{\theta_0}$ satisfies Axiom~\ref{ax:gen-basics-main-text}; $\succsim_{\theta_0}$ and $\{\succsim_{\theta_0,I}\}_{I\in \mathcal{I}}$ satisfy Axiom \ref{ax:asymptotic-benchmark-dominance}; and each $\succsim_{\theta_0,I}$ satisfies Axioms \ref{ax:gs-main-text}-\ref{ax:I-Caution-Main-Text}.
    \item[(ii)] $\succsim_{\theta_0}$ has a generalized $\alpha$-LAM risk representation on $\mathcal R_H$.
\end{itemize}
\end{theorem}

The axioms in Theorem \ref{thm-main-result} constitute a necessary and sufficient set of conditions on researcher behavior such that they are merely \emph{transparently consistent} with conventional efficiency arguments. In this sense, Theorem \ref{thm-main-result} serves as a useful benchmark towards eventually selecting a functional form for LAM risk. However, as we have previously discussed, transparent consistency by itself is not a sufficiently demanding desideratum, since it allows for preferences which violate some basic tenets of rational decision-making. As Example \ref{ex:violations} suggests, these violations are due to the arbitrary dependence of $\alpha(R)$ on $R\in \mathcal{R}_H$. This motivates the focus of the rest of this note on $\alpha$-LAM, which restores the basic axioms we want and brings us one step closer to pinning down a functional form for LAM risk.

\subsection{Remaining axioms for $\boldsymbol{\alpha}$-LAM}
Taking stock, recall that any risk representation of $\succsim_{\theta_0}$ must take a stance on how to aggregate risk across $n\geq 1$ and $h\in H$, and axioms allow us to translate this stance into concrete implications for researcher behavior. In particular, the functional form for $\alpha$-LAM risk summarizes risk via three aggregators---$\sup_{I\in\mathcal{I}}$, $\alpha \liminf_{n\to\infty}+(1-\alpha)\limsup_{n\to\infty}$, and $\max_{h\in I}$---that determine the criterion's attitudes towards uncertainty about $I\in\mathcal{I}$, $n\geq 1$, and $h\in I$, respectively. We have already imposed axioms relating $\succsim_{\theta_0}$ and $\{\succsim_{\theta_0,I}\}_{I\in\mathcal{I}}$ which ensure that the researcher takes the \emph{worst-case} over $I\in \mathcal{I}$ (Axiom \ref{ax:cons-and-caution-main-text}), as well as axioms on each $\succsim_{\theta_0,I}$ which ensure that the researcher takes the \emph{worst-case} over $h\in I$ (Axioms \ref{ax:gs-main-text}-\ref{ax:I-Caution-Main-Text}). It remains to provide axioms on each $\succsim_{\theta_0,I}$ which ensure that the researcher's $I$-conditional attitudes towards uncertainty about $n\geq1$ are described by the aggregator $\alpha \liminf_{n\to\infty}+(1-\alpha)\limsup_{n\to\infty}$ for some $\alpha \in [0,1]$. Recall that we may impose these axioms on the induced $I$-conditional preference $\succsim_{\theta_0,I}$, which ranks \emph{$I$-local risk function sequences} $\mathcal{R}_I$.

\paragraph{Axioms for uncertainty towards $\boldsymbol{n}$.}
We begin with a basic rationality condition which requires \emph{monotonicity} in sample size conditional on $I$.

\begin{axiom}[$I$-Sample Monotonicity]\label{ax:-sample-size-mono-main-text}
For each $R_I,R_I' \in \mathcal{R}_I$,
\[
(R_I)_n \succsim_{\theta_0,I} (R_I')_n \quad \forall n\geq 1 \implies R_I \succsim_{\theta_0,I} R_I'
\]
\end{axiom}

Suppose $R_I$ and $R_I'$ are induced at $\theta_0$ by $\delta$ and $\delta'$, respectively. Axiom \ref{ax:-sample-size-mono-main-text} requires that if the researcher $I$-conditionally prefers the estimator $\delta_n$ to the estimator $\delta_n'$ at every sample size $n\geq1$, then they $I$-conditionally prefer the entire estimator sequence $\delta$ to the entire estimator sequence $\delta'$.

To isolate the researcher's attitudes towards uncertainty about $n\geq 1$, we focus on the researcher's ranking of elements of elements of $\mathcal{R}_I$ which are \emph{constant in $h$}. Formally (as defined in Example \ref{ex:violations}), a \emph{local risk sequence} is a bounded sequence $K=(k_n)_{n\geq 1}$, where each $k_n\geq 0$. Let $B_+^{\mathbb{N}} \subseteq \mathcal{R}_H$ denote the set of local risk sequences.\footnote{By identifying each $K \in B_+^{\mathbb{N}}$ with the local risk function sequence $(\overrightarrow{k}_n)_{n\geq 1} \in \mathcal{R}_H$, we may view $B_+^{\mathbb{N}}$ as a subset of $\mathcal{R}_H$. Hence, any binary relation on $\mathcal{R}_H$ induces a binary relation on $B_+^{\mathbb{N}}$.} Since constant local risk functions may be generated by equivariant-in-law estimators, a local risk sequence may be generated by a sequence of equivariant-in-law estimators whose (constant in $h$) risk may depend on the sample index $n$. For example, the estimator sequences $\gamma$ defined in Example \ref{ex:violations} for bounded sequences $(a_n)_n$ induce local risk sequences.

We therefore impose our axioms on $\succsim_{\theta_0,I}$ restricted to the domain $B_+^{\mathbb{N}}$. Among other representations, \cite{marinacci1998axiomatic} studies representations of the form $K\mapsto \liminf_{n\to\infty} K_n$. We build upon the analysis in \cite{marinacci1998axiomatic} Theorem~7 to offer axioms for representations of the more general form $K\mapsto \alpha \liminf_{n\to\infty} K_n+(1-\alpha) \limsup_{n\to\infty} K_n$ for some $\alpha \in [0,1]$.

Stating these axioms requires several more definitions. Say that $K,K'\in B_+^{\mathbb{N}}$ are \emph{comonotonic} if $(K_n-K_m)(K'_n-K'_m)\geq 0$ for all $m,n\in\mathbb{N}$, and say that a set of local risk sequences is \emph{pairwise comonotonic} if every pair from it is comonotonic. In words, comonotonic sequences exhibit the same ordering of risks across $n$, which implies that any mixture retains this ordering. A useful special case is when $K$ and $K'$ are each monotone decreasing in $n$. In Example \ref{ex:gaussian}, the estimator sequence $\gamma_n$ defined above yields a local risk sequence which is monotone decreasing in $n$ when $|a_n|\downarrow 0$. Next, for each $K \in B_+^{\mathbb{N}}$ and $m \in \mathbb{N}$, define $K^m = (K_{m+1},K_{m+2},\ldots)$ to be the local risk sequence which truncates the first $m$ risks in $K$. In Example \ref{ex:gaussian}, if $K$ is induced at $\theta_0$ by $\gamma$, $K^m$ is induced at $\theta_0$ by the estimator sequence $(\gamma_n^m)_{n\geq 1}$, where each $\gamma_n^m=\Bar{X}_n+a_{m+n}/\sqrt{n}$. Finally, for each $K \in B_+^{\mathbb{N}}$ and each bijection $\pi:\mathbb{N}\rightarrow \mathbb{N}$, define $K^\pi$ to be the local risk sequence satisfying $K^\pi_n = K_{\pi(n)}$ for all $n\in\mathbb{N}$. If $K$ is induced at $\theta_0$ by an equivariant-in-law estimator sequence $\delta$, $K^\pi$ is induced at $\theta_0$ by an equivariant-in-law estimator sequence which, for each sample $n$, reproduces the constant local risk incurred by $\delta_{\pi(n)}$. In Example \ref{ex:gaussian}, if $K$ is induced by
$\gamma_n=\Bar X_n+a_n/\sqrt n$, then $K^\pi$ is induced by $\gamma_n^\pi
=\Bar X_n+a_{\pi(n)}/\sqrt n$.

\begin{axiom}[\citealt{marinacci1998axiomatic} Axioms]\label{ax:mar-98-main-text}
$\succsim_{\theta_0,I}$ on $B_+^{\mathbb{N}}$ satisfies:
\begin{itemize}
    \item[(i)] \textnormal{Nontrivial Weak Order:} $\succsim_{\theta_0,I}$ is nontrivial, complete, and transitive on $B_+^{\mathbb{N}}$.
    \item[(ii)] \textnormal{Archimedean Continuity:} For each $K,K',K'' \in B_+^{\mathbb{N}}$, if $K \succ_{\theta_0,I} K'$ and $K' \succ_{\theta_0,I} K''$, then there exist $\beta,\gamma \in (0,1)$ such that
    \[
    \beta K + (1-\beta)K'' \succ_{\theta_0,I} K'
    \quad\text{and}\quad
    K' \succ_{\theta_0,I} \gamma K + (1-\gamma)K''
    \]
    \item[(iii)] \textnormal{Comonotonic Independence:} For all pairwise comonotonic $K,K',K'' \in B_+^{\mathbb{N}}$ and all $\beta \in (0,1)$,
    \[
    K \succsim_{\theta_0,I} K' \implies \beta K + (1-\beta)K''
    \succsim_{\theta_0,I}
    \beta K' + (1-\beta)K''
    \]
    \item[(iv)] \textnormal{Truncation Invariance:} For each $K \in B_+^{\mathbb{N}}$ and $m \in \mathbb{N}$, $K \sim_{\theta_0,I} K^m$.
    \item[(v)] \textnormal{Permutation Invariance:} For each $K \in B_+^{\mathbb{N}}$ and each bijection $\pi:\mathbb{N}\rightarrow \mathbb{N}$, $K \sim_{\theta_0,I} K^\pi$.
\end{itemize}
\end{axiom}

Axiom \ref{ax:mar-98-main-text}(i)--(ii) are standard regularity conditions. Axiom \ref{ax:mar-98-main-text}(iii) requires that rankings are preserved under randomization with a common alternative when the local risk sequences being compared and the alternative order risk across sample sizes in the same way. For example, suppose $K$, $K'$, and $K''$ are induced at $\theta_0$ by $\gamma$, $\gamma'$, and $\gamma''$ defined as in Example \ref{ex:violations}, where each sequence of local biases is monotone decreasing in $n$. Axiom \ref{ax:mar-98-main-text}(iii) requires that if the researcher prefers $\delta$ to $\delta'$, then they prefer randomizing between $\delta$ and $\delta''$ to randomizing between $\delta'$ and $\delta''$ (where the probability of $\delta''$ is the same). Axioms \ref{ax:mar-98-main-text}(iv)-(v) enforce the asymptotic nature of the criterion, requiring that the researcher is indifferent to truncating and permuting risk sequences.

\paragraph{Characterizing $\boldsymbol{\alpha}$-LAM.} We may now use the axioms we have discussed to characterize $\alpha$-LAM.

\begin{theorem}[$\alpha$-LAM]\label{thm:fixed-alpha-LAM}
The following are equivalent:
\begin{itemize}
    \item[(i)] $\succsim_{\theta_0}$ satisfies Axiom~\ref{ax: basics-main-text}; $(\succsim_{\theta_0},\{\succsim_{\theta_0,I}\}_{I\in\mathcal I})$ satisfy Axiom~\ref{ax:cons-and-caution-main-text}; and each $\succsim_{\theta_0,I}$ satisfies Axioms~\ref{ax:gs-main-text}-\ref{ax:mar-98-main-text}.
    \item[(ii)] $\succsim_{\theta_0}$ has an $\alpha$-LAM risk representation on $\mathcal R_H$ for some $\alpha\in[0,1]$.
\end{itemize}
\end{theorem}

\paragraph{Characterizing classical and attainment LAM.} To understand what additional axioms select classical and attainment LAM from this class, we consider a motivating example.

\addtocounter{example}{-1}
\begin{example}[Continued]
\label{ex:uncertainty-affinity}
\normalfont Fix any centering value $\theta_0$. Let $A=\{10^k: k\in \mathbb{N}\}$, and define the estimator sequences $\delta,\delta'$ as:
\[
\delta_n(X_n)=
\begin{cases}\Bar{X}_n,& n\in A,\\\Bar{X}_n+1/\sqrt{n},& \text{ else},\end{cases}
\qquad
\delta_n'(X_n)=
\begin{cases}\Bar{X}_n+1/\sqrt{n},&n\in A,\\\Bar{X}_n,& \text{ else}\end{cases}
\]
In words, $\delta$ outputs the sample mean on sample sizes which are powers of 10 and otherwise outputs the sample mean with constant local bias $1$, while $\delta'$ does the opposite. The induced local risk sequences are:
\[
K_n=
\begin{cases}1,&n\in A,\\2,& \text{ else},\end{cases}
\qquad
K_n'=
\begin{cases}2,&n\in A,\\1,&\text{ else}\end{cases}
\]
The mixture $(K+K')/2$ yields the constant local risk sequence $(3/2,3/2,\ldots)$. Assume $\succsim_{\theta_0}$ has an $\alpha$-LAM representation. Then,
\[
V_{\alpha\text{-}\mathrm{LAM}}(K)
=V_{\alpha\text{-}\mathrm{LAM}}(K')=\alpha+(1-\alpha)2
\quad \text{and} \quad
V_{\alpha\text{-}\mathrm{LAM}}((K+K')/2)=3/2
\]
and hence
\[
\frac{1}{2}K+\frac{1}{2}K' \succsim_{\theta_0} K \sim_{\theta_0} K' \iff \alpha \leq 1/2
\]
Hence, comparing mixtures of two local risk sequences to each puts restrictions on the value of $\alpha$. In particular, $(K+K')/2 \succ_{\theta_0} K$ for $\alpha=0$ (attainment LAM) and $K \succ_{\theta_0} (K+K')/2$ for $\alpha=1$ (classical LAM).
\end{example}

Example \ref{ex:uncertainty-affinity} suggests that the key property of preferences which distinguishes classical LAM risk from attainment LAM risk is the researcher's attitudes towards \emph{hedging against uncertainty about $n$}. We may therefore state the sample-size-analog of Axiom \ref{ax:gs-main-text}(v), as well as its opposite.

\begin{axiom}[Sample Uncertainty Aversion]\label{ax:sample-size-uncertainty-aversion}
For each $K,K'\in B_+^{\mathbb N}$,
\[
K\sim_{\theta_0}K'
\implies
\beta K+(1-\beta)K'\succsim_{\theta_0}K
\quad\forall\beta\in(0,1).
\]
\end{axiom}

\begin{axiom}[Sample Uncertainty Affinity]\label{ax:sample-size-uncertainty-affinity}
For each $K,K'\in B_+^{\mathbb N}$,
\[
K\sim_{\theta_0}K'
\implies
K\succsim_{\theta_0}\beta K+(1-\beta)K'
\quad\forall\beta\in(0,1).
\]
\end{axiom}

Sample Uncertainty Aversion requires the researcher to prefer hedging between indifferent local risk sequences and selects $\alpha=0$, while Sample Uncertainty Affinity requires the researcher to dislike hedging between indifferent local risk sequences and selects $\alpha=1$.

\begin{corollary}[Classical and attainment LAM]\label{cor:CLAM-and-ALAM}
Suppose $\succsim_{\theta_0}$ has an $\alpha$-LAM risk representation on $\mathcal{R}_H$.
\begin{itemize}
    \item[(i)] $\succsim_{\theta_0}$ satisfies Axiom \ref{ax:sample-size-uncertainty-aversion} if and only if $\succsim_{\theta_0}$ is represented by $\overline{V}_{\mathrm{LAM}}$ on $\mathcal{R}_H$.
    \item[(ii)] $\succsim_{\theta_0}$ satisfies Axiom \ref{ax:sample-size-uncertainty-affinity} if and only if $\succsim_{\theta_0}$ is represented by $\underline{V}_{\mathrm{LAM}}$ on $\mathcal{R}_H$.
\end{itemize}
\end{corollary}

As the running example demonstrates, Sample Uncertainty Aversion seems normatively appealing: when two procedures are indifferent, randomizing between them hedges their sample-size-dependent risks. We therefore uniquely recommend attainment LAM as the functional form for LAM risk. Corollary~\ref{cor:CLAM-and-ALAM} provides a precise axiomatic foundation for this recommendation.

\setstretch{0.8}
\small 
\setlength{\bibsep}{0pt}
\bibliography{WorksCited}

\newpage 

\setstretch{1.2}
\normalsize 
\appendix

\section{Appendix}
\label{app:proofs}

\subsection{Discussion of LAN}
\label{app:LAN}
In this section, we provide a brief exposition of the \emph{local asymptotic normality} \eqref{eq:LAN} condition in Assumption \ref{as:1} for unfamiliar readers. For a more detailed treatment, see \cite{van2000asymptotic} Chapter 7. The key idea is to observe that the LHS of \eqref{eq:LAN} is a \emph{log-likelihood ratio} in the local experiment $\mathcal{P}_{\theta_0,n,H}$ that measures the relative strength of evidence in favor of the local parameter $h_n$ vs. $0$. Note that the log-likelihood ratio process of the Gaussian shift experiment $\mathcal{G}_{\theta_0,H}$ is 
\[
\log \frac{dN(h,I_{\theta_0}^{-1})}{dN(0,I_{\theta_0}^{-1})}(X)=h' I_{\theta_0}X-\frac{1}{2}h'I_{\theta_0}h
\]
where at the parameter $0$, $I_{\theta_0}X \sim N(0,I_{\theta_0})$. \eqref{eq:LAN} therefore ensures that, for sufficiently large $n$, the distribution of $\mathcal{P}_{\theta_0,n,H}$'s log-likelihood ratio process under the law $P_{\theta_0,n}$ approximates the distribution of $\mathcal{G}_{\theta_0,H}$'s log-likelihood ratio process under the law $N(0,I_{\theta_0}^{-1})$, for any finite subset of $H$. Hence, the limiting statistical difficulty of distinguishing parameters $\theta_0+h_n/\sqrt{n}$ and $\theta_0$ based on the experiment $\mathcal{P}_{\theta_0,n,H}$ equals the statistical difficulty of distinguishing parameters $h$ and $0$ based on the Gaussian shift experiment $\mathcal{G}_{\theta_0,H}$. It is in this sense that the sequence of local experiments $\mathcal{P}_{\theta_0,n,H}$ converges to the limit experiment $\mathcal{G}_{\theta_0,H}$.\footnote{More generally, they converge as statistical experiments in the sense of Le Cam (see, e.g., \citealt{van2000asymptotic} Definition 9.1).}

In particular, at the parameter $0$ in $\mathcal{G}_{\theta_0,H}$, $X \sim N(0,I_{\theta_0}^{-1})$ and the log-likelihood ratio is distributed 
\[
\log \frac{dN(h,I_{\theta_0}^{-1})}{dN(0,I_{\theta_0}^{-1})}(X) \sim N\left(-\frac{1}{2}h'I_{\theta_0}h,h'I_{\theta_0}h\right)
\]
which has negative mean and nontrivial variance for any $h\neq 0$. In words, the evidence is (correctly) in favor of parameter $0$ on average, but the difficulty of distinguishing $h$ and $0$ is nontrivial. The limiting stabilization in problem difficulty is obtained by letting the parameter $\theta_0+h_n/\sqrt{n}$ approach $\theta_0$ at the same rate as sampling-based uncertainty vanishes.

\subsection{Proofs of results in Section~\ref{sec:setup}}
\label{app-sec-2-proofs}

\begin{proof}[Proof of Proposition \ref{prop:stable-local-risks}]
Let $\delta$ be an estimator sequence satisfying \eqref{eq:stable-local-risks} for some $r_{\theta_0,\delta}: H\rightarrow \overline{\mathbb{R}}_+$. Fix any $I\in\cI$. There are two cases. First, suppose $\max_{h\in I} r_{\theta_0,\delta}(h)<+\infty$. By \eqref{eq:stable-local-risks}, $\max_{h\in I} R_{\theta_0,\delta}(n,h)<+\infty$ eventually for sufficiently large $n$. Since $I$ is finite, \eqref{eq:stable-local-risks} implies
\[
\left|
\max_{h\in I}R_{\theta_0,\delta}(n,h)
-
\max_{h\in I}r_{\theta_0,\delta}(h)
\right|
\leq
\max_{h\in I}
\left|R_{\theta_0,\delta}(n,h)-r_{\theta_0,\delta}(h)\right|
\to 0.
\]
Hence, 
\[
\limsup_{n\to\infty} \sup_{h \in I} R_{\theta_0,\delta}(n,h)=\liminf_{n\to\infty} \sup_{h \in I} R_{\theta_0,\delta}(n,h)=\max_{h\in I} r_{\theta_0,\delta}(h)
\]
Second, suppose $r_{\theta_0,\delta}(h)=+\infty$ for some $h\in I$. By \eqref{eq:stable-local-risks}, $R_{\theta_0,\delta}(n,h)\leq \max_{h'\in I} R_{\theta_0,\delta}(n,h') \to +\infty=\max_{h\in I} r_{\theta_0,\delta}(h)$, so
\[
\limsup_{n\to\infty} \sup_{h \in I} R_{\theta_0,\delta}(n,h)=\liminf_{n\to\infty} \sup_{h \in I} R_{\theta_0,\delta}(n,h)=\max_{h\in I} r_{\theta_0,\delta}(h)=+\infty
\]
Finally, taking the supremum over $I\in\cI$ gives
\[
\CVLAM(\delta)
=
\overline{V}_{\mathrm{LAM},\theta_0}(\delta)
=
\sup_{I\in\cI}\max_{h\in I}r_{\theta_0,\delta}(h)=\sup_{h\in H}r_{\theta_0,\delta}(h)
\]
since $\cI$ contains every singleton. For each transparently consistent $V_{\theta_0}$,
\[
\CVLAM(\delta)\leq V_{\theta_0}(\delta)\leq \AVLAM(\delta)
\]
we conclude that each inequality in the displayed equation above binds, as desired.

\end{proof}

\begin{proof}[Proof of Corollary \ref{cor:regular-stable-losses}]
By Theorem~2.20 of \cite{van2000asymptotic}, for each $h \in H$,
\[
R_{\theta_0,\delta}(n,h)=E_{\theta_0,n,h}[\ell(Z_{\theta_0,\delta,n}(h))]
\to
E\!\left[\ell\!\left(Z_{\theta_0,\delta}(h)\right)\right]
\]
Conclude by applying Proposition~\ref{prop:stable-local-risks}.
\end{proof}

\begin{proof}[Proof of Observation \ref{obs:gen-alpha-LAM-iff-obv-cons}]
Part (i) follows since, for any $\alpha_{\theta_0}: \mathcal{D} \rightarrow [0,1]$,
\[
\CVLAM(\delta)\leq V_{\mathrm{gen}\text{-}\alpha\text{-}\mathrm{LAM},\theta_0}(\delta) \leq \AVLAM(\delta) \quad \forall \delta \in \mathcal{D}
\]
For part (ii), let $V_{\theta_0}$ be a transparently consistent risk index at $\theta_0$. Define 
\[
\alpha_{\theta_0}(\delta)=\begin{cases} 
      \frac{\AVLAM(\delta)-V_{\theta_0}(\delta)}{\AVLAM(\delta)-\CVLAM(\delta)} & \CVLAM(\delta)<\AVLAM(\delta)<+\infty \\
      \frac{1}{2} & \CVLAM(\delta)=\AVLAM(\delta)
\end{cases}
\]
and extend it arbitrarily to $\mathcal{D}$.
\end{proof}

\subsection{Anscombe--Aumann environment}\label{app:AA}
In Appendix \ref{app:AA}, we study preferences in the decision environment of \emph{Anscombe--Aumann acts}, which is a standard decision environment in the microeconomic decision theory literature on choice under uncertainty. The main goal of this section is to provide a set of standard conditions on these more primitive preferences which allow us to study the induced preferences over \emph{utility act sequences}. This eventually allows us to derive preferences in the main text's decision environment of \emph{local risk function sequences}. For the sake of completeness, we also begin by studying preferences which do not condition on a centering value $\theta_0 \in \Theta$. Once we obtain the induced preferences over utility act sequences, we derive the $\theta_0$-conditional preferences studied in the main text.

\paragraph{Anscombe--Aumann act sequences.} Since the estimand takes values in $\mathbb{R}^m$, let $E=\mathbb{R}^m$ be the set of possible values of estimation errors. Let $Y$ be the set of Borel probability distributions (lotteries) over $E$. Let $Y^s$ be the set of \emph{simple} (finitely-supported) lotteries over $E$. For $e \in E$, let $D[e]$ be the Dirac lottery on $e$. An \emph{act} is a function $f: \Theta \times H \rightarrow Y$. An \emph{act sequence} is a function $f: \mathbb{N}\times \Theta \times H \rightarrow Y$. Let $\mathcal{F}_{AA}$ be the set of all act sequences. 

\begin{example-app}
\label{ex:induced-AA-acts}
\normalfont Recall that $Z_{\theta_0,\delta,n}(h)$ is the \emph{local estimation error} of estimator sequence $\delta$ at $\theta_0\in \Theta$, $n\geq 1$, and $h \in H$. Let $\mathcal{L}_\delta(n,\theta_0,h) \in Y$ denote the law of $Z_{\theta_0,\delta,n}(h)$ under $X_n \sim P_{\theta_0,n,h}$. Let $f_\delta \in \mathcal{F}_{AA}$ be the \emph{act sequence induced by $\delta$}, defined as $f_\delta(n,\theta_0,h)=\mathcal{L}_\delta(n,\theta_0,h)$. $f_\delta$ maps each $\theta_0,n,h$ to the distribution over local estimation errors induced by $\delta_n$ at the centering value $\theta_0$, when the local parameter is $h$ and the data is drawn according to $P_{\theta_0,n,h}$.
\end{example-app}

We take as primitive a binary relation $\succsim$ on $\mathcal{F}_{AA}$, which we interpret as the researcher's preferences over act sequences (which are collections of distributions over estimation errors indexed by the sample size, centering value, and local parameter). For each $p \in Y$, let $\overrightarrow{p} \in \mathcal{F}_{AA}$ be the act sequence defined as: $\overrightarrow{p}(n,\theta,h)=p$ for all $(n,\theta,h)$.\footnote{By identifying each $p \in Y$ with $\overrightarrow{p} \in \mathcal{F}_{AA}$, we may view $Y$ as a subset of $\mathcal{F}_{AA}$.} 

\begin{definition-app}
\begin{itemize}
    \item[(i)] An act sequence $f \in \mathcal{F}_{AA}$ is \emph{$\succsim$-bounded} if there exist $p,q \in Y^s$ such that
    \[
    \overrightarrow{p}\succsim \overrightarrow{f(n,\theta,h)} \ \ \text{ and }\ \ \overrightarrow{f(n,\theta,h)}\succsim \overrightarrow{q} \quad \forall (n,\theta,h)
    \]
    \item[(ii)] An act sequence $f \in \mathcal{F}_{AA}$ is \emph{$\succsim$-measurable} if, for each $p \in Y^s$, the following set is measurable:
    \[
    \Big\{(n,\theta,h) \in \mathbb{N}\times \Theta \times H: \overrightarrow{p}\succsim \overrightarrow{f(n,\theta,h)}\Big\}
    \]
\end{itemize}
\end{definition-app}

Let $\mathcal{F}_{AA}(\succsim)$ be the set of $\succsim$-bounded, $\succsim$-measurable, $Y^s$-valued act sequences. If the restriction of $\succsim$ to $Y^s$ is a weak order, $Y^s\subseteq \mathcal{F}_{AA}(\succsim)$. 

\paragraph{Utility act sequences.} A \emph{utility act sequence} is a bounded, measurable function $F:\mathbb N\times\Theta\times H\to\mathbb R_-$. Let $\mathcal F$ be the set of utility act sequences. A \emph{local utility act sequence} is a bounded, measurable function $F:\mathbb N\times H\to\mathbb R_-$. Let $\mathcal F_H$ denote the set of local utility act sequences.\footnote{By identifying each $F \in \mathcal{F}_H$ with the constant-in-$\theta$ utility act sequence $(n,\theta,h)\mapsto F(n,h)$, we may view $\mathcal{F}_H$ as a subset of $\mathcal{F}$.} For each $\theta_0\in\Theta$ and $F \in \mathcal{F}$, let $F_{\theta_0} \in \mathcal{F}_H$ be the \emph{local version of $F$ at $\theta_0$}: $F_{\theta_0}(n,h):=F(n,\theta_0,h)$. 

\paragraph{LAM representations on act sequences.} Let $\mathcal{I}$ be the set of nonempty finite subsets of $H$.

\begin{definition-app}
\label{defn:gen-alpha-LAM-AA}
\begin{itemize}
    \item[(i)] $\succsim$ has a \emph{generalized $\alpha$-LAM utility representation} on $\mathcal F_{AA}(\succsim)$ if there exist an onto, affine function $u:Y^s\to(-\infty,0]$ with $u(D[0])=0$ and a function $a:\mathcal F_H\to[0,1]$ such that $u$ represents the restriction of $\succsim$ to $Y^s$, and for each $f,g\in\mathcal F_{AA}(\succsim)$,
    \[
    f\succsim g
    \iff
    U_a((u\circ f)_{\theta_0})
    \geq
    U_a((u\circ g)_{\theta_0})
    \quad\forall\theta_0\in\Theta
    \]
    where for $F \in \mathcal{F}_H$:
    \begin{align*}
    \label{eq:gen-alpha-LAM-AA}
        U_a(F)=a(F) \inf_{I\in \mathcal{I}}\limsup_{n\to\infty} \min_{h\in I} F(n,h)+(1-a(F)) \inf_{I\in \mathcal{I}}\liminf_{n\to\infty} \min_{h\in I} F(n,h)
    \end{align*}
    \item[(ii)] $\succsim$ has an \emph{$\alpha$-LAM utility representation} on $\mathcal F_{AA}(\succsim)$ if it has a generalized $\alpha$-LAM utility representation on $\mathcal F_{AA}(\succsim)$ where $a$ is constant.
\end{itemize}
\end{definition-app}

\addtocounter{example-app}{-1}
\begin{example-app}[Continued]
\normalfont Suppose $(u,a)$ is a generalized $\alpha$-LAM representation of $\succsim$ on $\mathcal{F}_{AA}(\succsim)$, where $u(p)=-\int \ell \ dp$. Hence,
\[
-(u\circ f_{\delta})_{\theta_0}(n,h)=-u(f_\delta(n,\theta_0,h))=E_{\theta_0,n,h} \ell(Z_{\theta_0,\delta,n})=R_{\theta_0,\delta}(n,h)
\]
as defined in Definition \ref{defn:lrfs}. Hence, for each $\delta,\delta' \in \mathcal{D}$ with $f_{\delta},f_{\delta'} \in \mathcal{F}_{AA}(\succsim)$,
\[
f_{\delta} \succsim f_{\delta'} \iff V_{\mathrm{gen}-\alpha-\mathrm{LAM},\theta_0}(\delta') \geq V_{\mathrm{gen}-\alpha-\mathrm{LAM},\theta_0}(\delta) \quad \forall \theta_0 \in \Theta
\]
as defined in Definition \ref{defn-gen-alpha-LAM-estimators}, with $\alpha_{\theta_0}(\delta)=a(-R_{\theta_0,\delta})$. In this sense, the preference over estimator sequences expressed in Definition \ref{defn:gen-alpha-LAM-AA}(i) is the unanimity rule established by the set of preferences over estimator sequences expressed in Definition \ref{defn-gen-alpha-LAM-estimators}, where $\delta$ is \emph{unconditionally} preferred to $\delta'$ if it achieves lower generalized $\alpha$-LAM risk local to $\theta_0$, \emph{for every centering value $\theta_0$}.
\end{example-app}

Note that the function $u$ in an $\alpha$-LAM utility representation is unique up to positive scaling: if $(u,\alpha)$ represents $\succsim$ on $\mathcal F_{AA}(\succsim)$, then $(cu,\alpha)$ does as well for any $c>0$. Also note that if $(u,a)$ is a generalized $\alpha$-LAM representation of $\succsim$ on $\mathcal{F}_{AA}(\succsim)$, then since $u$ represents $\succsim$ on $Y^s$, $Y^s \subseteq \mathcal F_{AA}(\succsim)$. Since for each $p\in Y^s$, $U_a(\overrightarrow{u(p)})=u(p)$, the second requirement is consistent with the first.

\begin{axiom-app}[Unbounded EU]\label{ax:unbounded-eu}
$\succsim$ restricted to $Y^s$ satisfies:
\begin{itemize}
    \item[(i)] \textnormal{Weak Order:} $\succsim$ is complete and transitive on $Y^s$.
    \item[(ii)] \textnormal{Independence:} For each $p,q,r \in Y^s$ and $\alpha \in (0,1)$, $p\succ q$ implies $\alpha p+(1-\alpha)r \succ \alpha q+(1-\alpha)r$.
    \item[(ii)] \textnormal{Mixture Continuity:} For each $p,q,r \in Y^s$ with $p \succ q \succ r$, the following sets are open in $[0,1]$:
    \[
    \{\alpha \in [0,1]: \alpha p+(1-\alpha)r\succ q\} \quad \text{ and } \quad \{\alpha \in [0,1]: q\succ \alpha p+(1-\alpha)r\}
    \]
    \item[(iv)] \textnormal{Unboundedness From Below:} There exist $p\succ q$ in $Y^s$ such that, for each $\alpha \in (0,1)$, there exists $r \in Y^s$ with $q \succ \alpha r+(1-\alpha)p$.
    \item[(v)] \textnormal{Zero is best:} $D[0]\succsim D[z]$ for all $z\in E$.
\end{itemize}
\end{axiom-app}

\begin{lemma-app}
\label{lem:eu}
$\succsim$ restricted to $Y^s$ satisfies Unbounded EU if and only if it has an onto, affine utility representation $u: Y^s \rightarrow (-\infty,0]$ with $u(D[0])=0$.
\end{lemma-app}

\begin{proof}[Proof of Lemma \ref{lem:eu}]
Necessity is straightforward, so we prove sufficiency. By the main result of \cite{herstein1953axiomatic}, there exists an affine utility representation $v: Y^s\rightarrow \mathbb{R}$.\footnote{See, e.g., Theorem 5.27 of \cite{strzalecki2023decision}.} By Axiom \ref{ax:unbounded-eu}(v), $v(p)\leq v(D[0])$ for all $p\in Y^s$. By Axiom \ref{ax:unbounded-eu}(iv), $v(Y^s)$ is unbounded below. Since $v$ is affine, $v(Y^s)$ is convex, so $v(Y^s)=(-\infty,v(D[0])]$. Define $u:=v-v(D[0])$.
\end{proof}

The next axiom ensures that $\succsim$ induces a well-defined binary relation on the domain of utility act sequences.

\begin{axiom-app}[Utility Equivalence]\label{ax:mono}
$\succsim$ is transitive on $\mathcal F_{AA}(\succsim)$ and, for each $f,g\in\mathcal F_{AA}(\succsim)$,
\[
\overrightarrow{f(n,\theta,h)}\sim\overrightarrow{g(n,\theta,h)}
\quad\forall(n,\theta,h) \implies f\sim g.
\]
\end{axiom-app}

Whenever $\succsim$ satisfies Lemma \ref{lem:eu}, for each onto, affine function $u:Y^s\to(-\infty,0]$ with $u(D[0])=0$ which represents $\succsim$ on $Y^s$, define the binary relation $\succsim^u$ on $\mathcal F$ as:
\[
F\succsim^u G \iff f\succsim g
\text{ for some }f,g\in\mathcal F_{AA}(\succsim)
\text{ with }u\circ f=F\text{ and }u\circ g=G.
\]
Utility Equivalence ensures that the definition of $\succsim^u$ does not depend on the particular preimages $f,g$. We may define the analog of Definition \ref{defn:gen-alpha-LAM-AA} for $\succsim^u$. 

\begin{definition-app}
\label{defn:alpha-LAM-util}
\begin{itemize}
    \item[(i)] $\succsim^u$ has a \emph{generalized $\alpha$-LAM utility representation} on $\mathcal{F}$ if there exists a function $a:\mathcal{F}_H \rightarrow [0,1]$ such that, for each $F,G\in\mathcal F$,
    \[
    F\succsim^u G
    \iff
    U_a(F_{\theta_0})
    \geq
    U_a(G_{\theta_0})
    \quad\forall\theta_0\in\Theta
    \]
    \item[(ii)] $\succsim^u$ has an \emph{$\alpha$-LAM utility representation} on $\mathcal F$ if it has a generalized $\alpha$-LAM utility representation on $\mathcal F$ where $a$ is constant.
\end{itemize}
\end{definition-app}

The goal of Appendix \ref{app:AA} is to reduce the axiomatic exercise to studying $\succsim^u$, the induced preference over utility act sequences. To this end, the following lemma adapts standard arguments in the decision theory literature to our setting.\footnote{See, e.g., the arguments between Lemmas 1 and 2 of \cite{ball2023simple}.}

\begin{lemma-app}\label{lem:AA}
Fix any $a: \mathcal{F}_H\rightarrow [0,1]$ and onto, affine function $u:Y^s\to(-\infty,0]$ with $u(D[0])=0$. The following are equivalent:
\begin{itemize}
    \item[(i)] $\succsim$ satisfies Axiom \ref{ax:mono}, $u$ represents $\succsim$ on $Y^s$, and $a$ is a generalized $\alpha$-LAM utility representation of $\succsim^u$ on $\mathcal F$.
    \item[(ii)] $(u,a)$ is a generalized $\alpha$-LAM utility representation of $\succsim$ on $\mathcal F_{AA}(\succsim)$.
\end{itemize}
\end{lemma-app}

\begin{proof}[Proof of Lemma \ref{lem:AA}]
\underline{Forwards direction}: suppose $\succsim$ and $(u,a)$ satisfy part (i). Then, $u$ represents $\succsim$ on $Y^s$ by assumption. Fix any $f,g\in\mathcal F_{AA}(\succsim)$ and define $f_u=u\circ f$ and $g_u=u\circ g$. By definition of $\mathcal F_{AA}(\succsim)$ and since $u$ represents $\succsim$ on $Y^s$, $f_u,g_u\in \mathcal{F}$. Note that
\[
f\succsim g
\iff
f_u\succsim^u g_u.
\]
where the forwards implication follows from the definition of $\succsim^u$ and the backwards direction follows from Utility Equivalence. Since $a$ represents $\succsim^u$ on $\mathcal{F}$,
\[
f\succsim g
\iff
    U_a((f_u)_{\theta_0})
\geq
    U_a((g_u)_{\theta_0})
\quad\forall\theta_0\in\Theta
\]
Hence, $(u,a)$ represents $\succsim$ on $\mathcal{F}_{AA}(\succsim)$.

\underline{Backwards direction}: suppose $\succsim$ and $(u,a)$ satisfy part (ii). By definition, $u$ represents $\succsim$ on $Y^s$. $\succsim$ is transitive on $\mathcal{F}_{AA}(\succsim)$ because it has a utility representation on $\mathcal{F}_{AA}(\succsim)$. To verify the rest of Utility Equivalence, suppose $\overrightarrow{f(n,\theta,h)}\sim\overrightarrow{g(n,\theta,h)}$ for all $(n,\theta,h)$. Then, $f_u=g_u$ implies $U_a((f_u)_{\theta_0})=U_a((g_u)_{\theta_0})$ for all $\theta_0\in \Theta$ implies $f \sim g$. Hence, Utility Equivalence holds.

It remains to show that $a$ is a generalized $\alpha$-LAM utility representation of $\succsim^u$ on $\mathcal F$. Fix any $F,G\in\mathcal F$, and choose $c<0$ with $c\leq \min\{\inf F,\inf G\}$. Since $u$ is onto, we may choose $q\in Y^s$ with $u(q)=c$. For $t\in[c,0]$, set
\[
p_t:=\frac{c-t}{c}D[0]+\frac{t}{c}q.
\]
Note that $p_t\in Y^s$ and $u(p_t)=t$. Define $f,g\in \mathcal{F}_{AA}$ as $f(n,\theta,h):=p_{F(n,\theta,h)}$ and $g(n,\theta,h):=p_{G(n,\theta,h)}$. Then, $f,g\in \mathcal F_{AA}(\succsim)$ and satisfy $u\circ f=F$ and $u\circ g=G$. Hence,
\[
F\succsim^u G
\iff f \succsim g \iff U_a(F_{\theta_0})
\geq
U_a(G_{\theta_0})
\quad\forall\theta_0\in\Theta.
\]
where the first equivalence is due to Utility Equivalence and the definition of $\succsim^u$, and the second equivalence is because $(u,a)$ represents $\succsim$ on $\mathcal F_{AA}(\succsim)$ and $f_u=F,g_u=G$.
\end{proof}

Since the function $u$ in an $\alpha$-LAM utility representation is unique up to positive scaling, the following lemma follows by an analogous argument.

\begin{lemma-app}\label{lem:AA-fixed-alpha}
Fix any $\alpha \in [0,1]$. The following are equivalent:
\begin{itemize}
    \item[(i)] $\succsim$ satisfies Axioms~\ref{ax:unbounded-eu} and \ref{ax:mono}, and \emph{for each} onto, affine function $u:Y^s\to(-\infty,0]$ with $u(D[0])=0$ which represents $\succsim$ on $Y^s$, $\alpha$ is an $\alpha$-LAM utility representation of $\succsim^u$ on $\mathcal F$.
    \item[(ii)] There exists $u$ such that $(u,\alpha)$ is an $\alpha$-LAM utility representation of $\succsim$ on $\mathcal F_{AA}(\succsim)$.
\end{itemize}
\end{lemma-app}

\begin{proof}[Proof of Lemma \ref{lem:AA-fixed-alpha}]
\underline{Forwards direction}: suppose part (i) holds. By Unbounded EU, there exists an onto, affine function $u:Y^s\to(-\infty,0]$ with $u(D[0])=0$ which represents $\succsim$ on $Y^s$. By assumption, $\alpha$ is an $\alpha$-LAM utility representation of $\succsim^u$ on $\mathcal{F}$. By an exactly analogous argument as the proof of the forwards direction of Lemma \ref{lem:AA}, $(u,\alpha)$ represents $\succsim$ on $\mathcal{F}_{AA}(\succsim)$.

\underline{Backwards direction}: suppose part (ii) holds and let $(u,\alpha)$ represent $\succsim$ on $\mathcal{F}_{AA}(\succsim)$. Unbounded EU and Utility Equivalence follow from previous arguments. By cardinal uniqueness of EU, $v: Y^s\rightarrow (-\infty,0]$ is an onto, affine function with $v(D[0])=0$ which represents $\succsim$ on $Y^s$ if and only if $v=cu$ for some $c>0$. Fix any such $v$: by positive-scaling uniqueness of $\alpha$-LAM, $(v,\alpha)$ also represents $\succsim$. By an exactly analogous argument as the proof of the backwards direction of Lemma \ref{lem:AA}, $\alpha$ represents $\succsim^v$.
\end{proof}

With Lemma \ref{lem:AA} in hand, assume that $\succsim$ satisfies Axioms~\ref{ax:unbounded-eu} and \ref{ax:mono}, fix an onto, affine $u: Y^s\rightarrow (-\infty,0]$ with $u(D[0])=0$ which represents $\succsim$ on $Y^s$, and define $\succsim^u$ as above. We take $\succsim^u$ as primitive and seek necessary and sufficient conditions on $\succsim^u$ such that it has a generalized $\alpha$-LAM utility representation on $\mathcal{F}$, as well as necessary and sufficient conditions on $\succsim^u$ such that it has an $\alpha$-LAM utility representation on $\mathcal{F}$.\footnote{For the latter case, Lemma \ref{lem:AA-fixed-alpha} ensures that the choice of $u$ does not matter for the axiomatic exercise.} Since the rest of the Appendix works with utility act sequences, with abuse of notation, we henceforth denote $\succsim^u$ by $\succsim$ (unless otherwise stated).

\subsection{Utility act sequence environment}
Throughout this subsection, unless otherwise stated, all references to main-text axioms mean their utility-act versions under the bijection $R=-F$. Hence, the domains $\mathcal R_H$, $\mathcal R_I$, $\mathbb R_+^I$, and $B_+^{\mathbb N}$ are matched with $\mathcal F_H$, $\mathcal F_I$, $\mathbb R_-^I$, and $B_-^{\mathbb N}$, respectively. In axioms involving $F_n$, the function $F_n$ is identified with the utility act sequence that repeats that function at every sample size.

\label{app:util-act-lemmas}
\paragraph{Unanimity.} Recall that, for each $\theta_0 \in \Theta$ and $F \in \mathcal{F}$, we defined $F_{\theta_0} \in \mathcal{F}_H$ to be the local version of $F$ at $\theta_0$. We identify $F_{\theta_0} \in \mathcal{F}_H$ with the constant-in-$\theta$ utility act sequence in $\mathcal{F}$ defined as $(n,\theta,h)\mapsto F(n,\theta_0,h)$. Next, we use $\succsim$ to derive preferences over utility act sequences in $\mathcal{F}$ \emph{conditional} on each centering value $\theta_0 \in \Theta$. For each $\theta_0 \in \Theta$, define the binary relation $\succsim_{\theta_0}$ on $\mathcal{F}$ as:
\[
F \succsim_{\theta_0} G \iff F_{\theta_0} \succsim G_{\theta_0}
\]

Note that for each $F\in\mathcal F$ and $\theta_0,\theta\in\Theta$, $(F_{\theta_0})_\theta=F_{\theta_0}$, and hence for any $a: \mathcal{F}_H\rightarrow [0,1]$, $U_a((F_{\theta_0})_{\theta})=U_a(F_{\theta_0})$, which does not depend on $\theta$.

\begin{axiom-app}[$\Theta$-Unanimity]\label{ax:Theta-unanimity}
For each $F,G \in \mathcal{F}$,
\[
F \succsim G \iff F \succsim_{\theta_0} G \quad \forall \theta_0 \in \Theta
\]
\end{axiom-app}

Axiom \ref{ax:Theta-unanimity} requires that $F$ is preferred to $G$ if and only if the local version of $F$ at $\theta_0$ is preferred to the local version of $G$ at $\theta_0$, unanimously across $\theta_0 \in \Theta$. Analogously to Definition \ref{defn:alpha-LAM-util}, say that $\succsim_{\theta_0}$ has a \emph{generalized $\alpha$-LAM utility representation} on $\mathcal F$ if there exists $a: \mathcal{F}_H\rightarrow [0,1]$ such that the function $F\mapsto U_a(F_{\theta_0})$ represents $\succsim_{\theta_0}$ on $\mathcal{F}$, and say that $\succsim_{\theta_0}$ has an \emph{$\alpha$-LAM utility representation} on $\mathcal F$ if it has a generalized $\alpha$-LAM utility representation where $a$ is constant. The following lemma shows that $\Theta$-Unanimity allows us to fix a centering value $\theta_0\in \Theta$ and study the conditional preference $\succsim_{\theta_0}$.

\begin{lemma-app}\label{lem:unanimity-char}
For any $a: \mathcal{F}_H\rightarrow [0,1]$, the following are equivalent:
\begin{itemize}
    \item[(i)] $(\succsim,\{\succsim_{\theta_0}\}_{\theta_0\in\Theta})$ satisfy Axiom~\ref{ax:Theta-unanimity} and, for each $\theta_0\in \Theta$, $F\mapsto U_a(F_{\theta_0})$ is a generalized $\alpha$-LAM utility representation of $\succsim_{\theta_0}$ on $\mathcal{F}$.  
    \item[(ii)] $a$ is a generalized $\alpha$-LAM utility representation of $\succsim$ on $\mathcal F$.
\end{itemize}
\end{lemma-app}

\begin{proof}[Proof of Lemma \ref{lem:unanimity-char}]
\underline{Backwards direction}: suppose part (ii) holds. For each $\theta_0 \in \Theta$,
\[
F\succsim_{\theta_0}G
\iff F_{\theta_0}\succsim G_{\theta_0}
\iff U_a(F_{\theta_0})\geq U_a(G_{\theta_0})
\]
by definition of $\succsim_{\theta_0}$ and by the identity $U_a((F_{\theta_0})_{\theta})=U_a(F_{\theta_0})$. $\Theta$-Unanimity then follows by definition of generalized $\alpha$-LAM utility representation of $\succsim$ on $\mathcal{F}$.

\underline{Forwards direction}: suppose part (i) holds. Then,
\[
F\succsim G
\iff
F\succsim_{\theta_0}G\quad\forall\theta_0 \in \Theta
\iff
U_{a}(F_{\theta_0})\geq U_{a}(G_{\theta_0})\quad\forall\theta_0 \in \Theta
\]
by $\Theta$-Unanimity and by assumption that $a$ represents each $\succsim_{\theta_0}$ on $\mathcal{F}$.
\end{proof}

Note that by choosing constant $a$, the exact analog of Lemma~\ref{lem:unanimity-char} holds for $\alpha$-LAM utility representations. 

\begin{lemma-app}
\label{lem:common-generalized-selector}
The following are equivalent:
\begin{itemize}
    \item[(i)] For each $\theta_0 \in \Theta$, there exists $a_{\theta_0}: \mathcal{F}_H\rightarrow [0,1]$ such that $a_{\theta_0}$ is a generalized $\alpha$-LAM utility representation of $\succsim_{\theta_0}$ on $\mathcal{F}$.
    \item[(ii)] There exists $a: \mathcal{F}_H\rightarrow [0,1]$ such that, for each $\theta_0 \in \Theta$, $a$ is a generalized $\alpha$-LAM utility representation of $\succsim_{\theta_0}$ on $\mathcal{F}$.
\end{itemize}
\end{lemma-app}

\begin{proof}[Proof of Lemma \ref{lem:common-generalized-selector}]
The backwards direction is immediate. For the forwards direction, fix any $\theta_0 \in \Theta$, set $a:=a_{\theta_0}$, and note that
\[
F_{\theta_0} \succsim G_{\theta_0} \iff F \succsim_{\theta_0} G \iff U_{a}(F_{\theta_0}) \geq U_{a}(G_{\theta_0}) \quad \forall F,G \in \mathcal{F}
\]
This implies that $U_a$ represents $\succsim$ restricted to the subset of $\mathcal{F}$ which are constant in $\theta$. Hence, for any $\theta_0' \in \Theta$,
\[
F \succsim_{\theta_0'} G \iff F_{\theta_0'} \succsim G_{\theta_0'} \iff U_a(F_{\theta_0'})\geq U_a(G_{\theta_0'})
\]
as desired.
\end{proof}

Note that the exact analog of Lemma \ref{lem:common-generalized-selector} holds for constant $a$. With Lemma~\ref{lem:unanimity-char} and Lemma \ref{lem:common-generalized-selector} in hand, we fix $\theta_0\in\Theta$ and seek to axiomatically characterize the set of binary relations $\succsim_{\theta_0}$ with a generalized $\alpha$-LAM utility representation on $\mathcal{F}$, as well as the set of binary relations $\succsim_{\theta_0}$ with a $\alpha$-LAM utility representation on $\mathcal{F}$.

\paragraph{Restricting to $\boldsymbol{\mathcal{F}_H}$.} 
The following lemma shows that to achieve these goals, it suffices to study the restriction of  $\succsim_{\theta_0}$ to $\mathcal{F}_H$. 

\begin{lemma-app}\label{lem:F_H}
For any $a:\mathcal F_H\to[0,1]$, the following are equivalent.
\begin{itemize}
    \item[(i)] The function $U_a$ represents $\succsim_{\theta_0}$ on $\mathcal F_H$.
    \item[(ii)] The function $F\mapsto U_a(F_{\theta_0})$ represents $\succsim_{\theta_0}$ on $\mathcal F$.
\end{itemize}
\end{lemma-app}

\begin{proof}[Proof of Lemma \ref{lem:F_H}]
For the backwards direction, note that by identifying any $F\in \mathcal{F}_H$ with its constant-in-$\theta$ utility act sequence, $U_a(F_{\theta_0})=U_a(F)$. For the forwards direction, fix any $F,G \in \mathcal{F}$. Then,
\[
F\succsim_{\theta_0}G
\iff (F_{\theta_0})_{\theta_0}=F_{\theta_0} \succsim G_{\theta_0}=(G_{\theta_0})_{\theta_0} \iff F_{\theta_0}\succsim_{\theta_0}G_{\theta_0} \iff U_a(F_{\theta_0})\geq U_a(G_{\theta_0})
\]
by two applications of the definition of $\succsim_{\theta_0}$ and by part (i), since $F_{\theta_0},G_{\theta_0}\in\mathcal F_H$.
\end{proof}

Note that by choosing constant $a$, the exact analog of Lemma \ref{lem:F_H} holds for $\alpha$-LAM utility representations. Say that $\succsim_{\theta_0}$ has a \emph{generalized $\alpha$-LAM utility representation} on $\mathcal F_H$ if there exists $a: \mathcal{F}_H\rightarrow [0,1]$ such that $U_a$ represents $\succsim_{\theta_0}$ on $\mathcal F_H$, and say that $\succsim_{\theta_0}$ has an \emph{$\alpha$-LAM utility representation} on $\mathcal F_H$ if it has a generalized $\alpha$-LAM utility representation with constant $a$. Note that these representations are the equivalent utility-index analogs of Definition \ref{defn:LAM-risk-main-text}, and $\mathcal{F}_H$ is the utility-valued analog of the main text's decision environment $\mathcal{R}_H$. Henceforth, we fix $\theta_0 \in \Theta$ and seek to axiomatically characterize the set of binary relations $\succsim_{\theta_0}$ on $\mathcal{F}_H$ which have generalized $\alpha$-LAM utility representations, and which have $\alpha$-LAM utility representations.

\paragraph{Formal statement of Axiom \ref{as:2}.} We may motivate the axiomatic exercise described above directly in terms of the researcher's \emph{preference over estimator sequences}. Formally, in addition to the binary relation $\succsim_{\theta_0}$ on $\mathcal{F}_H$ (which we interpret as the researcher's preferences over local utility act sequences conditional on the centering value $\theta_0$), suppose we take as primitive the statistical decision environment defined in Section \ref{sec:setup} and the binary relation $\succsim_{\mathcal{D},\theta_0}$. We interpret $\succsim_{\mathcal{D},\theta_0}$ as the preferences over estimator sequences of the researcher, who faces the sequence of local informational environments and local decision problems at the centering value $\theta_0$ defined in Section \ref{sec:setup}.

\begin{axiom-app}[Local Risk Sufficiency]
\label{ax:cond-risk-relevance}
For each $\delta,\delta' \in \mathcal{D}$ with $R_{\theta_0,\delta},R_{\theta_0,\delta'} \in \mathcal{R}_H$,
\[
\delta \succsim_{\mathcal{D},\theta_0} \delta' \iff -R_{\theta_0,\delta} \succsim_{\theta_0} -R_{\theta_0,\delta'}
\]
\end{axiom-app}

Axiom \ref{ax:cond-risk-relevance} requires that performance in the sequence of local decision problems at $\theta_0$ is \emph{all that matters} for ranking estimator sequences conditional on $\theta_0$. It motivates our axiomatic exercise because, if $\succsim_{\theta_0}$ has a generalized $\alpha$-LAM utility representation on $\mathcal{F}_H$, then Local Risk Relevance holds if and only if $\succsim_{\mathcal{D},\theta_0}$ has a generalized $\alpha$-LAM risk representation $\mathcal{D}$ (for estimator sequences with bounded local risk function sequences at $\theta_0$, and where $a(\cdot)$ only depends on $\delta$ via $R_{\theta_0,\delta}$).

\begin{lemma-app}
\label{lem:eq:directed-alpha-identity}
Let 
\[
\overline{U}_I(F)=\limsup_{n\to\infty} \min_{h\in I} F(n,h) \quad \text{and} \quad \underline{U}_I(F)=\liminf_{n\to\infty} \min_{h\in I} F(n,h)
\]
For each $F \in \mathcal{F}_H$ and $t\in [0,1]$,
\begin{equation}\label{eq:directed-alpha-identity}
t \inf_{I \in \mathcal{I}} \overline{U}_I(F)+(1-t) \inf_{I \in \mathcal{I}} \underline{U}_I(F)=\inf_{I \in \mathcal{I}} \Big(t\overline{U}_I(F)+(1-t) \underline{U}_I(F) \Big)
\end{equation}
\end{lemma-app}

\begin{proof}[Proof of Lemma \ref{lem:eq:directed-alpha-identity}]
Fix any $F \in \mathcal{F}_H$ and $t\in [0,1]$. The $\leq$ direction is immediate. For the reverse inequality, fix $\epsilon>0$ and choose $\overline I,\underline I\in\mathcal I$ such that $\overline U_{\overline I}(F)<\inf_{I \in \mathcal{I}} \overline U_{I}(F)+\epsilon$ and $\underline U_{\underline I}(F)<\inf_{I \in \mathcal{I}} \underline U_{I}(F)+\epsilon$, and let $I=\overline I \cup \underline I$. Both $\overline U_I(F)$ and $\underline U_I(F)$ are weakly decreasing under set inclusion, so
\begin{align*}
    \inf_{J \in \mathcal{I}} \Big(t\overline{U}_J(F)+(1-t) \underline{U}_J(F) \Big)\leq t\overline{U}_I(F)+(1-t) \underline{U}_I(F) \\
    \leq t \overline U_{\overline I}(F)+(1-t) \underline U_{\underline I}(F) \leq t \inf_{I \in \mathcal{I}} \overline U_{I}(F)+(1-t)\inf_{I \in \mathcal{I}} \underline U_{I}(F)+\epsilon
\end{align*}
Conclude by taking $\epsilon\downarrow0$.
\end{proof}

\paragraph{Characterizing $\boldsymbol{\alpha}$-LAM.} It is convenient to begin with $\alpha$-LAM, since many of the intermediate lemmas prove useful for our characterization of generalized $\alpha$-LAM. In this section, we state a collection of lemmas which, when combined, axiomatically characterize the set of binary relations $\succsim_{\theta_0}$ on $\mathcal{F}_H$ which have an $\alpha$-LAM utility representation, hence proving Theorem \ref{thm:fixed-alpha-LAM}.

First, we show that Axiom \ref{ax: basics-main-text} ensures that $\succsim_{\theta_0}$ on $\mathcal{F}_H$ admits certainty equivalents. Say that a function $J: \mathcal{F}_H\rightarrow \mathbb{R}_-$ is \emph{normalized} if $J(\overrightarrow{k})=k$ for each $k\leq 0$.

\begin{lemma-app}\label{lem:CEs}
Suppose $\succsim_{\theta_0}$ satisfies Axiom \ref{ax: basics-main-text}. For each $F \in \mathcal{F}_H$, there exists a unique $J_{\theta_0}(F) \in \mathbb{R}_-$ such that $F \sim \overrightarrow{J_{\theta_0}(F)}$. Hence, $J_{\theta_0}: \mathcal{F}_H \rightarrow \mathbb{R}_-$ is a normalized utility representation of $\succsim_{\theta_0}$.
\end{lemma-app}

\begin{proof}[Proof of Lemma \ref{lem:CEs}]
Fix any $F \in \mathcal{F}_H$. Since $F$ is bounded, there exist $K,k \in \mathbb{R}_-$ such that $K \geq F(n,h) \geq k$ for all $(n,h)$. By monotonicity (and transitivity), $\overrightarrow{K} \succsim_{\theta_0} F \succsim_{\theta_0} \overrightarrow{k}$. By mixture continuity, the sets $P_F=\{\alpha \in [0,1]: \alpha \overrightarrow{K}+(1-\alpha)\overrightarrow{k}\succsim_{\theta_0} F\}$ and $L_F=\{\alpha \in [0,1]: F \succsim_{\theta_0} \alpha \overrightarrow{K}+(1-\alpha)\overrightarrow{k}\}$ are closed. By completeness, $P_F\cup L_F=[0,1]$. By above, $P_F$ and $L_F$ are nonempty, since $1 \in P_F$ and $0 \in L_F$. Since $[0,1]$ is connected, there exists $\alpha' \in P_F \cap L_F$. Define $J_{\theta_0}(F):=\alpha' K+(1-\alpha')k \in \mathbb{R}_-$. Uniqueness immediately follows from monotonicity and transitivity.
\end{proof}

\paragraph{Consistency and caution.}
For convenience, we state the utility act version of Definition~\ref{def:I-cond} here. For each $I\in\mathcal I$, define the \emph{$I$-conditional preference on $\mathcal{F}_H$} as:
\[
F\succsim_{\theta_0,I}G
\iff
F_I\overrightarrow{0}\succsim_{\theta_0}G_I\overrightarrow{0}
\]

\begin{lemma-app}\label{lem:I-rep}
Fix $\alpha\in[0,1]$. Suppose $\succsim_{\theta_0}$ is represented on $\mathcal F_H$ by
\[
U_\alpha(F)=\alpha \inf_{J\in\mathcal I} \limsup_{n\to\infty} \min_{h\in J}F(n,h)+(1-\alpha) \inf_{J\in\mathcal I} \liminf_{n\to\infty} \min_{h\in J}F(n,h)
\]
For each $I\in\mathcal I$, $\succsim_{\theta_0,I}$ is represented on $\mathcal F_H$ by
\[
U_{\alpha,I}(F)= \alpha \limsup_{n\to\infty} \min_{h\in I}F(n,h)+(1-\alpha) \liminf_{n\to\infty} \min_{h\in I}F(n,h)
\]
\end{lemma-app}

\begin{proof}[Proof of Lemma \ref{lem:I-rep}]
Fix any $F \in \mathcal{F}_H$ and $I\in\mathcal I$. By Lemma \ref{lem:eq:directed-alpha-identity}, we may write
\[
U_\alpha(F_I\overrightarrow{0})=\inf_{J\in \mathcal{I}} U_{\alpha,J}(F_I\overrightarrow{0})
\]
For each $J \in\mathcal I$, there are two cases. First, if $J$ and $I$ are disjoint, then $\min_{h\in J} (F_I\overrightarrow{0})(n,h)=0$ for all $n\geq 1$, so $U_{\alpha,J}(F_I \overrightarrow{0})=0\geq U_{\alpha,I}(F)$. Second, if $J\cap I\neq\varnothing$, then since $F\leq 0$,
\[
\min_{h\in J}(F_I\overrightarrow{0})(n,h)
=
\min_{h\in J\cap I}F(n,h)
\geq
\min_{h\in I}F(n,h)
\quad\forall n\geq 1
\]
which implies $U_{\alpha,J}(F_I\overrightarrow{0})\geq U_{\alpha,I}(F)$. In particular, setting $J=I$ gives equality. Hence, $U_\alpha(F_I\overrightarrow{0})=U_{\alpha,I}(F)$. The conclusion follows by definition of $\succsim_{\theta_0,I}$.
\end{proof}

The next lemma shows that Consistency and Caution characterize the $\inf_{I\in\mathcal I}$ aggregator and force the neighborhood-specific coefficients to equal one common $\alpha$.

\begin{lemma-app}\label{lem:consistency-caution}
For each $\theta_0\in\Theta$, the following are equivalent:
\begin{itemize}
    \item[(i)] $\succsim_{\theta_0}$ satisfies Axiom~\ref{ax: basics-main-text}, $(\succsim_{\theta_0},\{\succsim_{\theta_0,I}\}_{I\in\mathcal I})$ satisfy Axiom~\ref{ax:cons-and-caution-main-text}, and, for each $I\in\mathcal I$, $\succsim_{\theta_0,I}$ is represented on $\mathcal F_H$ by $U_{\alpha_I,I}$ for some $\alpha_I\in[0,1]$.
    \item[(ii)] For some $\alpha\in[0,1]$, $\succsim_{\theta_0}$ is represented on $\mathcal F_H$ by
    \[
    U_\alpha(F)=\inf_{I\in\mathcal I}U_{\alpha,I}(F).
    \]
\end{itemize}
Moreover, the coefficients $\{\alpha_I\}_I$ in part (i) are all equal to the coefficient $\alpha$ in part (ii).
\end{lemma-app}

\begin{proof}[Proof of Lemma \ref{lem:consistency-caution}]
\underline{Forwards direction.} Suppose part (i) holds. Since $\succsim_{\theta_0}$ satisfies Axiom~\ref{ax: basics-main-text}, Lemma~\ref{lem:CEs} implies the normalized utility representation $J_{\theta_0}$. Define $\widetilde U(F):=\inf_{I\in\mathcal I}U_{\alpha_I,I}(F)$. We have:
\begin{align*}
    U_{\alpha_I,I}(F)\geq\widetilde U(F) \quad \forall I\in \mathcal{I} \implies F \succsim_{\theta_0,I} \overrightarrow{\Tilde{U}(F)} \quad \forall I\in \mathcal{I} \\
    \implies F \succsim_{\theta_0} \overrightarrow{\Tilde{U}(F)} \implies J_{\theta_0}(F)\geq \Tilde{U}(F)
\end{align*}
since each $U_{\alpha_I,I}$ is normalized and by Consistency.

If $\widetilde U(F)=0$, then $J_{\theta_0}(F)\leq0=\widetilde U(F)$. If $\widetilde U(F)<0$, for each $\epsilon>0$ with $\epsilon<-\widetilde U(F)$, there exists $I\in \mathcal{I}$ such that $U_{\alpha_I,I}(F)<\widetilde U(F)+\epsilon$, which implies:
\begin{align*}
    \overrightarrow{\widetilde U(F)+\epsilon} \succ_{\theta_0,I} F \implies \overrightarrow{\widetilde U(F)+\epsilon} \succsim_{\theta_0} F \implies \widetilde U(F)+\epsilon\geq J_{\theta_0}(F)
\end{align*}
since each $U_{\alpha_I,I}$ is normalized and by Caution. Taking $\epsilon\downarrow0$ yields
\[
J_{\theta_0}(F)=\inf_{I\in\mathcal I}U_{\alpha_I,I}(F).
\tag{A.1}
\]
It remains to show that $\alpha_I=\alpha_J=:\alpha$ for all $I,J\in\mathcal{I}$.

Fix $I\in\mathcal I$. By definition of $\succsim_{\theta_0,I}$, the function $F\mapsto J_{\theta_0}(F_I\overrightarrow{0})$ represents $\succsim_{\theta_0,I}$ on $\mathcal{F}_H$. Furthermore, this function is normalized: for each $k\in \mathbb{R}_-$ and $J\in \mathcal{I}$, 
\[
\min_{h\in J}(\overrightarrow{k}_I\overrightarrow{0})(n,h)=\begin{cases} 
    0 \ \forall n\geq 1 & J\cap I=\emptyset \\
    k \ \forall n\geq 1 & \text{else}
\end{cases} \implies U_{\alpha_J,J}(\overrightarrow{k}_I\overrightarrow{0})=\begin{cases} 
    0 & J\cap I=\emptyset \\
    k & \text{else}
\end{cases}
\]
and hence $J_{\theta_0}(\overrightarrow{k}_I\overrightarrow{0})=k$ by (A.1). By assumption, $U_{\alpha_I,I}$ is also a (normalized) utility representation of $\succsim_{\theta_0,I}$ on $\mathcal{F}_H$. Since normalized utility representations are unique, 
\[
\inf_{J\in \mathcal{I}} U_{\alpha_J,J}(F_I\overrightarrow{0})=J_{\theta_0}(F_I\overrightarrow{0})=U_{\alpha_I,I}(F) \quad \forall F\in \mathcal{F}_H
\]
For $K\in B_-^{\mathbb N}$, an analogous computation as above yields:
\[
U_{\alpha_J,J}(K_I\overrightarrow{0})=\begin{cases} 
    0 & J\cap I=\emptyset \\
    \alpha_J \limsup_{n\to\infty} K_n+(1-\alpha_J) \liminf_{n\to\infty} K_n & \text{else}
\end{cases}
\]
Substituting $K=(-1,0,-1,0,\ldots)$ yields:
\[
\inf_{\substack{J\in\mathcal I\\J\cap I\neq\varnothing}}\alpha_J=\alpha_I
\]
This identity holds for each $I\in \mathcal{I}$. Finally, fix any $I,J\in \mathcal{I}$. There are two cases: if $I\cap J\neq \emptyset$, then $\alpha_J\geq \alpha_I$ and $\alpha_I \geq \alpha_J$, so $\alpha_I=\alpha_J$. Else, $I\cup J$ intersects both $I$ and $J$, so $\alpha_I=\alpha_{I\cup J}=\alpha_J$. Hence all $\alpha_I$ equal some common $\alpha$.

\underline{Backwards direction.} $\succsim_{\theta_0}$ satisfying Axiom~\ref{ax: basics-main-text} is immediate. By Lemma~\ref{lem:I-rep}, each $\succsim_{\theta_0,I}$ is represented by $U_{\alpha,I}$. Consistency and Caution follow directly from $U_\alpha=\inf_{I\in\mathcal I}U_{\alpha,I}$.
\end{proof}

With Lemma~\ref{lem:consistency-caution} in hand, we fix $I \in \mathcal{I}$ and axiomatically characterize the set of binary relations $\succsim_{\theta_0,I}$ on $\mathcal{F}_H$ which have a representation $U_{\alpha,I}$ for some $\alpha\in[0,1]$.

\paragraph{Restricting $\boldsymbol{\succsim_{\theta_0,I}}$ to $\boldsymbol{\mathcal{F}_I}$.} Endow $I$ with its power set $\sigma$-algebra. Let $\mathcal{F}_I$ be the set of bounded, measurable act sequences $F: \mathbb{N} \times I \rightarrow \mathbb{R}_-$. For each $F\in \mathcal{F}_H$, let $F_I \in \mathcal{F}_I$ be the restriction of $F$ to $I$. Next, we show it suffices to focus attention on the preference that $\succsim_{\theta_0,I}$ induces on $\mathcal{F}_I$.

\begin{lemma-app}\label{lem:properties-of-I-cond}
Suppose $\succsim_{\theta_0}$ is a weak order on $\mathcal{F}_H$. Then $\succsim_{\theta_0,I}$ satisfies:
\begin{itemize}
    \item[(i)] \emph{$I$-Transitivity}: $\succsim_{\theta_0,I}$ is transitive on $\mathcal F_H$.
    \item[(ii)] \emph{$I$-Relevance}: for each $F,G\in \mathcal{F}_H$, if $F_I=G_I$, then $F\sim_{\theta_0,I}G$.
\end{itemize}
\end{lemma-app}

\begin{proof}[Proof of Lemma \ref{lem:properties-of-I-cond}]
Part (i) follows from transitivity of $\succsim_{\theta_0}$. For part (ii), $F_I=G_I$ implies $F_I\overrightarrow{0}=G_I\overrightarrow{0}$, so reflexivity of $\succsim_{\theta_0}$ and the definition of $\succsim_{\theta_0,I}$ give $F\sim_{\theta_0,I}G$.
\end{proof}

Next, with abuse of notation, define the binary relation $\succsim_{\theta_0,I}$ on $\mathcal{F}_I$ as: 
\[
F \succsim_{\theta_0,I} G \iff \exists F',G' \in \mathcal{F}_H \text{ with } (F')_I=F, (G')_I=G  \text{ and } F' \succsim_{\theta_0,I} G'
\]

\begin{lemma-app}\label{lem:I-relevance}
Fix $\alpha\in[0,1]$. The following are equivalent:
\begin{itemize}
    \item[(i)] $\succsim_{\theta_0,I}$ satisfies $I$-Transitivity and $I$-Relevance on $\mathcal F_H$, and $U_{\alpha,I}$ represents $\succsim_{\theta_0,I}$ on $\mathcal F_I$.
    \item[(ii)] $U_{\alpha,I}$ represents $\succsim_{\theta_0,I}$ on $\mathcal F_H$.
\end{itemize}
\end{lemma-app}

\begin{proof}[Proof of Lemma \ref{lem:I-relevance}]
\underline{Forwards direction.} Under $I$-Relevance and $I$-Transitivity,
\[
F\succsim_{\theta_0,I}G
\iff
F_I\succsim_{\theta_0,I}G_I
\quad\text{for all }F,G\in\mathcal F_H.
\]
The forwards implication follows from the definition of $\succsim_{\theta_0,I}$ on $\mathcal F_I$. For the reverse implication, choose $F',G'\in\mathcal F_H$ with $(F')_I=F_I$ and $(G')_I=G_I$ such that $F' \succsim_{\theta_0,I} G'$. By $I$-Relevance and $I$-Transitivity, $F\succsim_{\theta_0,I}G$. Finally, note that $U_{\alpha,I}(F)$ depends only on $F_I$.

\underline{Backwards direction.} $\succsim_{\theta_0,I}$ satisfying $I$-Transitivity and $I$-Relevance is immediate. For $F,G\in\mathcal F_I$, extend them arbitrarily to $F',G'\in\mathcal F_H$. Then
\[
U_{\alpha,I}(F)\geq U_{\alpha,I}(G)
\iff
U_{\alpha,I}(F')\geq U_{\alpha,I}(G')
\iff
F'\succsim_{\theta_0,I}G',
\]
which, by above, is equivalent to $F\succsim_{\theta_0,I}G$.
\end{proof}

With Lemma~\ref{lem:I-relevance} in hand, we axiomatically characterize the set of binary relations $\succsim_{\theta_0,I}$ on $\mathcal{F}_I$ which have a representation $U_{\alpha,I}$ for some $\alpha \in [0,1]$. We first record another implication of $I$-Transitivity and $I$-Relevance.

\begin{lemma-app}\label{lem:I-transitivity}
Suppose $\succsim_{\theta_0,I}$ satisfies $I$-Transitivity and $I$-Relevance on $\mathcal{F}_H$. $\succsim_{\theta_0,I}$ is transitive on $\mathcal{F}_I$.
\end{lemma-app}

\begin{proof}[Proof of Lemma \ref{lem:I-transitivity}]
Fix $F,G,H \in \mathcal{F}_I$, and assume $F \succsim_{\theta_0,I} G$ and $G \succsim_{\theta_0,I} H$. By definition of $\succsim_{\theta_0,I}$ on $\mathcal{F}_I$, there exist $F',G' \in \mathcal{F}_H$ with $(F')_I=F,(G')_I=G$ and $F' \succsim_{\theta_0,I} G'$, and there exist $G'',H'' \in \mathcal{F}_H$ with $(G'')_I=G,(H'')_I=H$ and $G'' \succsim_{\theta_0,I} H''$. By $I$-Relevance, $(G')_I=G=(G'')_I$ implies $G' \sim_{\theta_0,I} G''$. By $I$-Transitivity, $F' \succsim_{\theta_0,I} H''$. By definition of $\succsim_{\theta_0,I}$ on $\mathcal{F}_I$, $F \succsim_{\theta_0,I} H$.
\end{proof}

Next, it is useful to study $\succsim_{\theta_0,I}$ restricted to two sub-choice domains of interest. First, let $\mathbb{R}_-^I$ denote the set of \emph{constant act sequences} in $\mathcal{F}_I$. Second, let $B_-^{\mathbb{N}}$ denote the set of \emph{sequences of constant acts} in $\mathcal{F}_I$.

\begin{definition-app}
$\succsim_{\theta_0,I}$ on $\mathbb{R}_-^I$ has a \emph{maxmin expected utility (MEU) representation with ambiguity set $\Delta(I)$} if the function
\[
U_{I,MEU}(f)=\min_{h\in I} f(h)
\]
represents $\succsim_{\theta_0,I}$ on $\mathbb{R}_-^I$.
\end{definition-app}

\begin{definition-app}
Fix $\alpha\in[0,1]$. The relation $\succsim_{\theta_0,I}$ has an \emph{$\alpha$-limit representation} on $B_-^{\mathbb N}$ if the function
\[
H_\alpha(K)=\alpha \limsup_{n\to\infty} K_n+(1-\alpha) \liminf_{n\to\infty} K_n
\]
represents $\succsim_{\theta_0,I}$ on $B_-^{\mathbb N}$.
\end{definition-app}

The next lemma combines the MEU representation over $h$ with an $\alpha$-limit representation over $n$.

\begin{lemma-app}\label{lem:LAM-iff-MEU-and-LS}
Fix $\alpha\in[0,1]$. The following are equivalent:
\begin{itemize}
    \item[(i)] $\succsim_{\theta_0,I}$ has an MEU representation with ambiguity set $\Delta(I)$ on $\mathbb R_-^I$; it has an $\alpha$-limit representation on $B_-^{\mathbb N}$; it satisfies Axiom~\ref{ax:-sample-size-mono-main-text}; and it is transitive on $\mathcal F_I$.
    \item[(ii)] $\succsim_{\theta_0,I}$ is represented by $U_{\alpha,I}$ on $\mathcal F_I$.
\end{itemize}
\end{lemma-app}

\begin{proof}[Proof of Lemma \ref{lem:LAM-iff-MEU-and-LS}]
For the forwards direction, let $k_{F_n}:=\min_{h\in I}F(n,h)$ and define $k_{G_n}$ analogously. The MEU representation implies
\[
F_n\sim_{\theta_0,I}\overrightarrow{k_{F_n}},
\qquad
G_n\sim_{\theta_0,I}\overrightarrow{k_{G_n}}
\quad\forall n\geq 1
\]
Axiom~\ref{ax:-sample-size-mono-main-text} implies
\[
F\sim_{\theta_0,I}(\overrightarrow{k_{F_n}})_n,
\qquad
G\sim_{\theta_0,I}(\overrightarrow{k_{G_n}})_n
\]
Hence, by transitivity on $\mathcal F_I$, the $\alpha$-limit representation therefore gives
\[
F\succsim_{\theta_0,I}G
\iff
H_\alpha((k_{F_n})_n)\geq H_\alpha((k_{G_n})_n)
\iff
U_{\alpha,I}(F)\geq U_{\alpha,I}(G).
\]
For the backwards direction, the restriction of $U_{\alpha,I}$ to $\mathbb R_-^I$ is $\min_{h\in I}$, and its restriction to $B_-^{\mathbb N}$ is $H_\alpha$. Monotonicity of $H_\alpha$ gives $I$-Sample Monotonicity, and the existence of a utility representation gives transitivity.
\end{proof}

With Lemma \ref{lem:LAM-iff-MEU-and-LS} in hand, we may characterize the two representations on the subdomains $\mathbb{R}_-^I$ and $B_-^{\mathbb{N}}$ separately. We use results from \cite{gilboa1989maxmin} to obtain the MEU representation on $\mathbb R_-^I$, and we build on results from \cite{marinacci1998axiomatic} to obtain the $\alpha$-limit representation on $B_-^{\mathbb N}$.

\begin{lemma-app}\label{lem-gs}
$\succsim_{\theta_0,I}$ on $\mathbb{R}_-^I$ satisfies Axioms \ref{ax:gs-main-text}-\ref{ax:I-Caution-Main-Text} if and only if it has a MEU representation with ambiguity set $\Delta(I)$.
\end{lemma-app}

\begin{proof}[Proof of Lemma \ref{lem-gs}]
Necessity of the axioms is straightforward, so we show sufficiency. By an exactly analogous argument as Lemma 6 of \cite{andrews2026misspecification}, Axiom \ref{ax:gs-main-text} implies that $\succsim_{\theta_0,I}=\geq$ on constant acts $\mathbb{R}_-$. Since the consequence set $\mathbb{R}_-$ is a convex subset of the vector space $\mathbb{R}$, the above implies that Axiom \ref{ax:gs-main-text} is equivalent to Axioms A.1, A.2', and A.3-A.6 of \cite{maccheroni2006ambiguity}. Hence, there exists a nonconstant, affine $u: \mathbb{R}_- \rightarrow \mathbb{R}$ and nonempty, closed, convex subset $C\subseteq \Delta(I)$ such that $\succsim_{\theta_0,I}$ on $\mathbb{R}_-^I$ is represented by 
\[
f\mapsto \min_{p\in C} \langle u(f),p\rangle
\]
By cardinal uniqueness, the function
\[
W_I(f)=\min_{p\in C} \langle f,p\rangle
\]
also represents $\succsim_{\theta_0,I}$ on $\mathbb{R}_-^I$. Note that $W_I$ is normalized. By Axiom \ref{ax:I-Caution-Main-Text}, for each $h \in I$, $f \in \mathbb{R}_-^I$, and $k\leq 0$, $W_I(f)\geq k$ implies $f(h)\geq k$, which implies for each $h \in I$, $f \in \mathbb{R}_-^I$, $f(h)\geq W_I(f)$, which implies $\min_{h\in I} f(h)\geq W_I(f)$. Conversely, since $C\subseteq \Delta(I)$, $\min_{h\in I} f(h)\leq W_I(f)$. Combining the two inequalities gives $W_I(f)=\min_{h\in I}f(h)$ for each $f\in\mathbb R_-^I$. Hence the normalized MEU index represents $\succsim_{\theta_0,I}$ on $\mathbb R_-^I$.
\end{proof}

\begin{lemma-app}\label{lem:marinacci-1998-rep}
Suppose $\succsim_{\theta_0,I}$ satisfies Axiom~\ref{ax:-sample-size-mono-main-text} and agrees with $\geq$ on constant utility acts. Its restriction to $B_-^{\mathbb N}$ satisfies Axiom~\ref{ax:mar-98-main-text} if and only if there exists a unique $\alpha_I\in[0,1]$ such that
\[
H_{\alpha_I}(K)
=
\alpha_I\limsup_{n\to\infty}K_n
+(1-\alpha_I)\liminf_{n\to\infty}K_n.
\]
represents $\succsim_{\theta_0,I}$ on $B_-^{\mathbb{N}}$.
\end{lemma-app}

The proof of Lemma~\ref{lem:marinacci-1998-rep} builds upon the Choquet expected utility representation obtained in Theorem~7 of \cite{marinacci1998axiomatic}. Removing Sample Uncertainty Affinity leaves one free parameter, which is $\alpha_I$.

\begin{proof}[Proof of Lemma \ref{lem:marinacci-1998-rep}]
Necessity is straightforward, so we show sufficiency. By standard and analogous arguments as before, $\succsim_{\theta_0,I}$ admits certainty equivalents on $B_-^{\mathbb{N}}$: for each $K \in B_-^{\mathbb{N}}$, there exists a unique $J_I(K)\leq0$ such that $K\sim_{\theta_0,I}\overrightarrow{J_I(K)}$. Note that $J_I$ is a normalized utility representation of $\succsim_{\theta_0,I}$ on $B_-^{\mathbb N}$.

Let $B_1^{\mathbb{N}}$ be the space of bounded sequences with values in $(-\infty,1]$, and define $\widehat J_I: B_1^{\mathbb{N}}\rightarrow (-\infty,1]$ as
\[
\widehat J_I(X):=J_I(X-\overrightarrow 1)+1
\]
The consequence interval $(-\infty,1]$ is convex and contains $[-1,1]$. Note that $\widehat J_I$ is normalized and inherits comonotonic independence and monotonicity from $J_I$, in addition to the other properties implied by Axiom~\ref{ax:mar-98-main-text}. By the Corollary of the main result in Section 3 of \cite{schmeidler1986integral}, there exists a unique normalized capacity $\nu_I$ on $2^{\mathbb N}$ such that
\[
\widehat J_I(X)=\int X\,d\nu_I
\quad\forall X\in B_1^{\mathbb{N}}
\]
In particular, we may apply this representation to indicator sequences of subsets $A\subseteq \mathbb{N}$. If $A$ is finite, a sufficiently long truncation of $\boldsymbol 1_A$ is zero, so Truncation Invariance gives $\nu_I(A)=0$. If $A$ is cofinite, a sufficiently long truncation is one, so $\nu_I(A)=1$. Finally, Permutation Invariance implies that all infinite coinfinite subsets $A,B\subseteq \mathbb N$ have the same capacity. To see this, we may enumerate $A$, $A^c$, $B$, and $B^c$ as:
\[
A=\{a_1<a_2<a_3<\cdots\}\quad
A^c=\{c_1<c_2<c_3<\cdots\}
\]
and
\[
B=\{b_1<b_2<b_3<\cdots\}\quad
B^c=\{d_1<d_2<d_3<\cdots\}
\]
Define $\pi:\mathbb N\to\mathbb N$ by
\[
\pi(a_k)=b_k
\quad\text{and}\quad
\pi(c_k)=d_k
\quad \forall k\in\mathbb N
\]
This is well-defined because $A$ and $A^c$ partition $\mathbb N$. It is straightforward to show that $\pi$ is one-to-one and onto. By construction, $n\in A$ iff $\pi(n)\in B$, so
\[
\mathbf 1[n\in A]=\mathbf 1[\pi(n)\in B]
\quad \forall n\geq 1
\]
Hence by Permutation Invariance,
\[
\mathbf 1_A \sim_{\theta_0,I} \mathbf 1_B
\]
which implies $\nu_I(A)=\nu_I(B)$. Let $c_I=\nu_I(A)$ for all infinite coinfinite $A\subseteq \mathbb{N}$. We have therefore shown
\[
\nu_I(A)
=
\begin{cases}
0&A\text{ finite} \\
c_I&A\text{ and }A^c\text{ infinite} \\
1 &A\text{ cofinite}
\end{cases}
\]
The formula for the Choquet integral now gives (again, see the Corollary in Section 3 of \cite{schmeidler1986integral}), for each $X \in B_1^{\mathbb{N}}$,
\begin{align*}
    \int X\,d\nu_I
    =\int_{x=-\infty}^0 \Big(\nu_I(X>x)-1\Big) \ dx+\int_{x=0}^1 \nu_I(X>x) \ dx \\
\end{align*}
Note that for almost all $x\in (-\infty,1]$,
\[
\{n\in \mathbb{N}: X_n>x\} \text{ is }\begin{cases} 
      \text{cofinite} & x< \liminf X \\
      \text{coinfinite and infinite} & \liminf X< x< \limsup X \\
      \text{finite} & x> \limsup X
   \end{cases}
\]
There are two cases. For $X$ with $\liminf X<0$,
\[
\int X \ d\nu_I=\int_{x=\liminf X}^0 (c_I-1) \ dx+\int_{x=0}^{\limsup X} c_I \ dx=(1-c_I)\liminf X+c_I \limsup X
\]
For $X$ with $\liminf X\geq 0$,
\[
\int X \ d\nu_I=\int_{x=0}^{\liminf X} 1 \ dx+\int_{x=\liminf X}^{\limsup X} c_I \ dx=(1-c_I)\liminf X+c_I\limsup X
\]
Hence, for $K\in B_-^{\mathbb N}$,
\[
J_I(K)
=
\widehat J_I(K+\overrightarrow 1)-1
=
(1-c_I)\liminf_{n\to\infty}K_n
+c_I\limsup_{n\to\infty}K_n
\]
Setting $\alpha_I:=c_I$ yields the desired representation $H_{\alpha_I}$. To see uniqueness, let $K=(-1,0,-1,0,\ldots)$. Then $H_\alpha(K)=\alpha-1$. If $H_{\alpha'}$ represents $\succsim_{\theta_0,I}$ on $B_-^{\mathbb{N}}$, then $\overrightarrow{\alpha-1} \sim_{\theta_0,I} \overrightarrow{\alpha'-1}$ and $\alpha=\alpha'$, since $H_\alpha,H_{\alpha'}$ are normalized and since $\succsim_{\theta_0,I}=\geq$ on constant acts.
\end{proof}

\paragraph{$\boldsymbol{\alpha}$-LAM representation theorem.} 
For the sake of completeness, we characterize the class of preferences on $\mathcal{F}$ with a $\theta_0$-unconditional $\alpha$-LAM utility representation, as defined in Definition \ref{defn:alpha-LAM-util}(ii). The $\theta_0$-conditional case in Theorem \ref{thm:fixed-alpha-LAM} immediately follows.

\begin{theorem-app}\label{app-thm-alpha-LAM-rep}
Let $\succsim$ be a binary relation on $\mathcal F$. The following are equivalent:
\begin{itemize}
    \item[(i)] $(\succsim,\{\succsim_{\theta_0}\}_{\theta_0\in\Theta})$ satisfy Axiom~\ref{ax:Theta-unanimity}; for each $\theta_0 \in \Theta$, $\succsim_{\theta_0}$ satisfies Axiom~\ref{ax: basics-main-text} and $(\succsim_{\theta_0},\{\succsim_{\theta_0,I}\}_{I\in\mathcal I})$ satisfy Axiom~\ref{ax:cons-and-caution-main-text}; and for each $\theta_0 \in \Theta$ and $I \in \mathcal{I}$, $\succsim_{\theta_0,I}$ satisfies Axioms~\ref{ax:gs-main-text}-\ref{ax:mar-98-main-text}.
    \item[(ii)] For some $\alpha\in[0,1]$, $\succsim$ has an $\alpha$-LAM utility representation on $\mathcal F$.
\end{itemize}
\end{theorem-app}

\begin{proof}[Proof of Theorem~\ref{app-thm-alpha-LAM-rep}]
\underline{Forwards direction.} Fix $\theta_0\in\Theta$ and consider the restriction of $\succsim_{\theta_0}$ to
$\mathcal F_H$. We establish a conditional representation before
aggregating across centering values. 

Fix $I\in\mathcal I$. By Weak Order and
Lemma~\ref{lem:properties-of-I-cond}, $\succsim_{\theta_0,I}$ satisfies $I$-Transitivity and $I$-Relevance on
$\mathcal F_H$. Consequently, its induced relation $\succsim_{\theta_0,I}$ on $\mathcal F_I$ is
well-defined and, by Lemma~\ref{lem:I-transitivity}, transitive.
Lemma~\ref{lem-gs} gives the representation $f \mapsto \min_{h\in I}f(h)$ on $\mathbb R_-^I$, which also gives the usual order $\geq$ on constant
utility acts. Together with $I$-Sample Monotonicity, the hypotheses of
Lemma~\ref{lem:marinacci-1998-rep} are therefore satisfied. This provides a unique $\alpha_{\theta_0,I}\in[0,1]$ such that $K \mapsto \alpha_{\theta_0,I}\limsup_{n\to\infty}K_n
 +(1-\alpha_{\theta_0,I})\liminf_{n\to\infty}K_n$ represents the restriction to $B_-^{\mathbb N}$. We have established all four conditions in
Lemma~\ref{lem:LAM-iff-MEU-and-LS}(i): the MEU representation, the
$\alpha_{\theta_0,I}$-limit representation, $I$-Sample Monotonicity,
and transitivity on $\mathcal F_I$. Hence,
$U_{\alpha_{\theta_0,I},I}$ represents $\succsim_{\theta_0,I}$ on
$\mathcal F_I$. Lemma~\ref{lem:I-relevance} extends this representation to
$\mathcal F_H$.

Since this conclusion holds for every $I \in \mathcal{I}$, we apply
Lemma~\ref{lem:consistency-caution}, which implies that all $\alpha_{\theta_0,I}$ equal a common
$\alpha_{\theta_0}\in[0,1]$ and that $F\mapsto \inf_{I\in\mathcal I}U_{\alpha_{\theta_0},I}(F)$ represents $\succsim_{\theta_0}$ on $\mathcal F_H$.
By Lemma~\ref{lem:eq:directed-alpha-identity}, this index equals
\[
 U_{\alpha_{\theta_0}}(F)
 =\alpha_{\theta_0}\inf_{I\in\mathcal I}
       \limsup_{n\to\infty}\min_{h\in I}F(n,h)
 +(1-\alpha_{\theta_0})\inf_{I\in\mathcal I}
       \liminf_{n\to\infty}\min_{h\in I}F(n,h)
\]
This proves the desired characterization at a fixed centering value.

Lemma~\ref{lem:F_H} now implies that
$F\mapsto U_{\alpha_{\theta_0}}(F_{\theta_0})$ represents
$\succsim_{\theta_0}$ on $\mathcal F$. By Lemma~\ref{lem:common-generalized-selector}, we may choose $\alpha\in[0,1]$ such that $\alpha=\alpha_{\theta_0}$ for each $\theta_0\in \Theta$. Finally, $\Theta$-Unanimity and
Lemma~\ref{lem:unanimity-char} give
\[
 F\succsim G
 \iff U_\alpha(F_{\theta_0})\geq U_\alpha(G_{\theta_0})
 \quad\text{for every }\theta_0\in\Theta
\]
as desired.

\underline{Backwards direction.} Suppose that $\alpha\in[0,1]$ gives the representation in part (ii).
Lemma~\ref{lem:unanimity-char} yields $\Theta$-Unanimity and the
representation $F\mapsto U_\alpha(F_{\theta_0})$ of every
$\succsim_{\theta_0}$ on $\mathcal F$.
Lemma~\ref{lem:F_H} therefore gives the representation $U_\alpha$ on
$\mathcal F_H$. By Lemma~\ref{lem:eq:directed-alpha-identity},
$U_\alpha=\inf_{I\in\mathcal I}U_{\alpha,I}$. Lemma~\ref{lem:consistency-caution} gives
the Basic Axioms and Consistency and Caution. 

For each $I$, Lemma~\ref{lem:I-rep} identifies the induced conditional
representation as $U_{\alpha,I}$, and Lemma~\ref{lem:I-relevance}
restricts it to $\mathcal F_I$. Lemma~\ref{lem:LAM-iff-MEU-and-LS} gives
the MEU and $\alpha$-limit representations and $I$-Sample Monotonicity.
Lemma~\ref{lem-gs} then gives the GS89 axioms and $I$-Maximal Caution,
and Lemma~\ref{lem:marinacci-1998-rep} gives the Marinacci axioms.
Each of these implications hold for every $\theta_0\in \Theta$ and every $I\in \mathcal{I}$, which shows part (i).
\end{proof}

\paragraph{Classical and attainment LAM representation theorems.}

\begin{proof}[Proof of Corollary \ref{cor:CLAM-and-ALAM}]
Suppose $\succsim_{\theta_0}$ is represented on $\mathcal{F}_H$ by
\[
U_{\alpha}(F)=\alpha \inf_{I\in \mathcal{I}} \limsup_{n\to\infty} \min_{h\in I} F(n,h)+(1-\alpha) \inf_{I\in \mathcal{I}} \liminf_{n\to\infty} \min_{h\in I} F(n,h)
\]

\underline{Part (i)}: Suppose $\alpha=0$, and fix any $K,K' \in B_-^{\mathbb{N}}$ with $K \sim_{\theta_0} K'$, which implies $\liminf_{n\to\infty} K_n=\liminf_{n\to\infty} K_n'$. Moreover,
\[
\liminf_{n\to\infty}\bigl(\beta K_n+(1-\beta)K_n'\bigr)
\geq
\beta\liminf_{n\to\infty}K_n+(1-\beta)\liminf_{n\to\infty}K_n'
\]
Thus Sample Uncertainty Aversion holds. Conversely, suppose Sample Uncertainty Aversion holds. Choose a partition of infinite sets $A,B,C$ of $\mathbb N$, and let $K=-\boldsymbol 1_{A\cup B}$ and $K'=-\boldsymbol 1_{A\cup C}$. Then we have
\[
U_\alpha(K)=U_\alpha(K')=-(1-\alpha), \ U_\alpha((K+K')/2)=\alpha (-1/2)+(1-\alpha)(-1)
\]
By Sample Uncertainty Aversion,
\[
U_\alpha((K+K')/2) \geq U_\alpha(K)
\]
so $\alpha=0$.

\underline{Part (ii)}: Suppose $\alpha=1$, and fix any $K,K' \in B_-^{\mathbb{N}}$ with $K \sim_{\theta_0} K'$, which implies $\limsup_{n\to\infty} K_n=\limsup_{n\to\infty} K_n'$. Moreover,
\[
\limsup_{n\to\infty}\bigl(\beta K_n+(1-\beta)K_n'\bigr)
\leq
\beta\limsup_{n\to\infty}K_n+(1-\beta)\limsup_{n\to\infty}K_n'
\]
Thus Sample Uncertainty Affinity holds. Conversely, suppose Sample Uncertainty Affinity holds. Choose a partition of infinite sets $A,B,C$ of $\mathbb N$, and let $K=-\boldsymbol 1_B$ and $K'=-\boldsymbol 1_C$. Then we have
\[
U_\alpha(K)=U_\alpha(K')=-(1-\alpha), \ U((K+K')/2)=(1-\alpha)(-1/2)
\]
By Sample Uncertainty Affinity,
\[
U((K+K')/2) \leq U_\alpha(K)
\]
so $\alpha=1$.
\end{proof}

\paragraph{Characterizing generalized $\boldsymbol{\alpha}$-LAM.} In this section, we axiomatically characterize the set of binary relations $\succsim_{\theta_0}$ on $\mathcal{F}_H$ which have a generalized $\alpha$-LAM utility representation by proving Theorem \ref{thm-main-result}. For this section only, it is convenient to revert to the language of risk used in the main text, such that $\succsim_{\theta_0}$ refers to the ranking on $\mathcal R_H$ studied in the main text. We begin by stating and proving some useful lemmas. Note that we may identify the subset of local risk function sequences $\mathcal{R}_H$ that are constant in $n\geq 1$ with the set $\mathbb{R}_+^H$ of bounded, measurable functions $R: H\rightarrow \mathbb{R}_+$.

\begin{lemma-app}
\label{lem:generalized-alpha-static-I-representation}
Suppose $\succsim_{\theta_0}$ has a generalized $\alpha$-LAM risk representation on $\mathcal R_H$ for some $a:\mathcal R_H\to[0,1]$. For each $I\in\mathcal I$, the restriction of $\succsim_{\theta_0,I}$ to $\mathbb{R}_+^H$ is represented by $M_I(R):=\max_{h\in I} R(h)$.
\end{lemma-app}

\begin{proof}
Fix $I\in\mathcal I$, and fix any $R,R'\in\mathbb R_+^H$. 
Observe that
\[
R\succsim_{\theta_0,I}R'
\iff
R_I\overrightarrow{0} \succsim_{\theta_0} R'_I\overrightarrow{0} \iff V_{\mathrm{gen}\text{-}\alpha\text{-}\mathrm{LAM}}(R'_I\overrightarrow{0})\geq V_{\mathrm{gen}\text{-}\alpha\text{-}\mathrm{LAM}}(R_I\overrightarrow{0})
\]
by definition of $\succsim_{\theta_0,I}$. Recall that, for each $\Tilde{R} \in \mathcal{R}_H$,
\[
V_{\mathrm{gen}\text{-}\alpha\text{-}\mathrm{LAM}}(\Tilde{R})
=a(\Tilde{R})
\sup_{J\in\mathcal I} \liminf_{n\to\infty}\max_{h\in J}\Tilde{R}(n,h)
+
\bigl(1-a(\Tilde{R})\bigr)
\sup_{J\in\mathcal I} \limsup_{n\to\infty}\max_{h\in J}\Tilde{R}(n,h)
\]
Let $\Tilde{R}:=R_I\overrightarrow{0}$.
Since $\Tilde{R}$ is constant in $n$, for each $J\in \mathcal{I}$ with $J\cap I$ nonempty,
\[
\liminf_{n\to\infty}\max_{h\in J}\Tilde{R}(n,h)=\limsup_{n\to\infty}\max_{h\in J}\Tilde{R}(n,h)=\max_{h\in J\cap I} R(h)
\]
and equals $0$ otherwise for $J$ disjoint from $I$. Hence,
\[
\sup_{J\in \mathcal{I}}\liminf_{n\to\infty}\max_{h\in J}\Tilde{R}(n,h)=\sup_{J\in \mathcal{I}}\limsup_{n\to\infty}\max_{h\in J}\Tilde{R}(n,h)=\max_{h\in I} R(h)
\]
which implies that $V_{\mathrm{gen}\text{-}\alpha\text{-}\mathrm{LAM}}(R_I\overrightarrow{0})
=\max_{h\in I} R(h)$. By an exactly analogous argument, $V_{\mathrm{gen}\text{-}\alpha\text{-}\mathrm{LAM}}(R_I'\overrightarrow{0})
=\max_{h\in I} R'(h)$, and we have shown
\[
R \succsim_{\theta_0,I} R' \iff \max_{h\in I} R'(h) \geq \max_{h\in I} R(h)
\]
as desired.
\end{proof}

\begin{lemma-app}
\label{lem:ABD-equivalence}
Suppose $\succsim_{\theta_0}$ satisfies Axiom \ref{ax:gen-basics-main-text} and, for each $I\in \mathcal{I}$, $\succsim_{\theta_0,I}$ satisfy Axioms \ref{ax:gs-main-text}-\ref{ax:I-Caution-Main-Text}. $(\succsim_{\theta_0},\{\succsim_{\theta_0,I}\}_{I\in \mathcal{I}})$ satisfies Axiom \ref{ax:asymptotic-benchmark-dominance} if and only if:
\begin{equation}
\label{eq:lam-bounds}
    \overrightarrow{\underline V_{\mathrm{LAM}}(R)}
    \succsim_{\theta_0}R
    \succsim_{\theta_0}
    \overrightarrow{\overline V_{\mathrm{LAM}}(R)}
    \qquad
    \text{for each }R\in\mathcal R_H \tag{LAM Bounds}
\end{equation}
\end{lemma-app}

\begin{proof}[Proof of Lemma \ref{lem:ABD-equivalence}]
For $R\in\mathcal R_H$ and $I\in\mathcal I$, write
\[
m_I^R(n):=\max_{h\in I}R(n,h),
\qquad
a_I(R):=\liminf_{n\to\infty}m_I^R(n),
\qquad
b_I(R):=\limsup_{n\to\infty}m_I^R(n),
\]
and
\[
L(R):=\sup_{I\in\mathcal I}a_I(R),
\qquad
B(R):=\sup_{I\in\mathcal I}b_I(R).
\]

By Lemma \ref{lem-gs}, each $\succsim_{\theta_0,I}$ is represented by the function $M_I(r)=\max_{h\in I} r(h)$ on $\mathbb{R}_+^I$. By an exactly analogous argument to Lemma \ref{lem:I-relevance}, each $\succsim_{\theta_0,I}$ is represented by the function $M_I(R)=\max_{h\in I} R(h)$ on $\mathbb{R}_+^H$, the set of bounded, measurable local risk function sequences that are constant in $n$.

\underline{Backwards direction}: suppose \eqref{eq:lam-bounds} holds. First, we show Axiom \ref{ax:asymptotic-benchmark-dominance}(i). Fix $R \in \mathcal{R}_H$ and $k\in \mathbb{R}_+$, and suppose that $R_n \succsim_{\theta_0,I} \overrightarrow{k}$ e.v. for all $I \in \mathcal{I}$. By above, this implies $k\geq M_I(R_n)$ e.v. for all $I \in \mathcal{I}$. Then,
\[
\overline V_{\mathrm{LAM}}(R)=\sup_{I\in \mathcal{I}} \limsup_{n\to\infty} M_I(R_n)\leq k
\]
and hence by Axiom \ref{ax:gen-basics-main-text} and \eqref{eq:lam-bounds},
\[
R \succsim_{\theta_0} \overrightarrow{\overline V_{\mathrm{LAM}}(R)} \succsim_{\theta_0} \overrightarrow{k}
\]
Similarly for Axiom \ref{ax:asymptotic-benchmark-dominance}(ii), suppose there exists $I\in\mathcal I$ such that $\overrightarrow{k} \succsim_{\theta_0,I} R_n$ e.v., which by above implies $M_I(R_n)\geq k$ e.v. Then,
\[
\underline V_{\mathrm{LAM}}(R)=\sup_{I\in \mathcal{I}} \liminf_{n\to\infty} M_I(R_n) \geq k
\]
and hence by Axiom \ref{ax:gen-basics-main-text} and \eqref{eq:lam-bounds},
\[
k \succsim_{\theta_0} \overrightarrow{\underline V_{\mathrm{LAM}}(R)} \succsim_{\theta_0} R
\]

\underline{Forwards direction}: suppose Axiom \ref{ax:asymptotic-benchmark-dominance} holds. We begin by showing
\begin{equation}
\label{eq:lower-bound-from-ABD}
\overrightarrow{\underline V_{\mathrm{LAM}}(R)}\succsim_{\theta_0}R
\end{equation}
There are two cases. If $\underline V_{\mathrm{LAM}}(R)=0$, then since $R(n,h)\geq0$ for each $(n,h)$, Axiom \ref{ax:asymptotic-benchmark-dominance}(ii) with $k=0$ gives \eqref{eq:lower-bound-from-ABD}. Suppose instead that $\underline V_{\mathrm{LAM}}(R)>0$. For each $m\geq2$, let
\[
k_m:=\left(1-\frac{1}{m}\right)\underline V_{\mathrm{LAM}}(R)
\]
For each $m\geq 2$, there exists $I_m\in\mathcal I$ such that $\liminf_{n\to\infty} M_{I_m}(R_n)>k_m$ and hence $M_{I_m}(R_n)\geq k_m$ eventually, which by Axiom \ref{ax:asymptotic-benchmark-dominance}(ii) implies $\overrightarrow{k_m}\succsim_{\theta_0}R $. Since this holds for each $m\geq 2$, the set
\[
\left\{
\beta\in[0,1]:
\beta\overrightarrow{\underline V_{\mathrm{LAM}}(R)}\succsim_{\theta_0}R
\right\}
\]
contains $1-1/m$ for each $m\geq2$. Since it is closed by Constant Mixture Continuity, it contains $1$, which proves \eqref{eq:lower-bound-from-ABD}.

Similarly, for each $m\geq 1$ let $j_m=\overline V_{\mathrm{LAM}}(R)+1/m$. For each $m\geq 1$ and $I\in\mathcal I$,
\[
\limsup_{n\to\infty} M_I(R_n)\leq \overline V_{\mathrm{LAM}}(R)<j_m \implies M_I(R_n)\leq j_m \text{ e.v.} \implies R\succsim_{\theta_0}\overrightarrow{j_m} 
\]
by Axiom \ref{ax:asymptotic-benchmark-dominance}(i). Since this holds for each $m\geq 1$, the set
\[
\left\{
\beta\in[0,1]:
R\succsim_{\theta_0}
\beta\overrightarrow{\overline V_{\mathrm{LAM}}(R)+1}+(1-\beta)\overrightarrow{\overline V_{\mathrm{LAM}}(R)}
\right\}
\]
contains $1/m$ for each $m\geq1$, and hence contains $0$ by Constant Mixture Continuity. Hence,
\begin{equation}
\label{eq:upper-bound-from-ABD}
R\succsim_{\theta_0}\overrightarrow{\overline V_{\mathrm{LAM}}(R)}
\end{equation}
Equations~\eqref{eq:lower-bound-from-ABD} and \eqref{eq:upper-bound-from-ABD} are \eqref{eq:lam-bounds}.
\end{proof}

\paragraph{Generalized $\boldsymbol{\alpha}$-LAM representation theorem for utility act sequences.} For the sake of completeness, we characterize the class of preferences on $\mathcal{F}$ with a $\theta_0$-unconditional generalized $\alpha$-LAM utility representation, as defined in Definition \ref{defn:alpha-LAM-util}(i). The $\theta_0$-conditional case in Theorem \ref{thm-main-result} immediately follows.

\begin{theorem-app}\label{app-thm-gen-alpha-LAM-rep}
Let $\succsim$ be a binary relation on $\mathcal F$. The following are equivalent:
\begin{itemize}
    \item[(i)] $(\succsim,\{\succsim_{\theta_0}\}_{\theta_0\in\Theta})$ satisfy Axiom~\ref{ax:Theta-unanimity}; for each $\theta_0 \in \Theta$, $\succsim_{\theta_0}$ satisfies Axiom~\ref{ax:gen-basics-main-text} and $(\succsim_{\theta_0},\{\succsim_{\theta_0,I}\}_{I\in\mathcal I})$ satisfy Axiom~\ref{ax:asymptotic-benchmark-dominance}; and for each $\theta_0\in \Theta$ and $I \in \mathcal{I}$, $\succsim_{\theta_0,I}$ satisfies Axioms~\ref{ax:gs-main-text} and \ref{ax:I-Caution-Main-Text}.
    \item[(ii)] $\succsim$ has a generalized $\alpha$-LAM utility representation on $\mathcal F$.
\end{itemize}
\end{theorem-app}

\begin{proof}[Proof of Theorem~\ref{app-thm-gen-alpha-LAM-rep}]
In this proof, we work in the language of risk and let $\succsim^r$ denote the preference over risk function sequences.

\underline{Forwards direction.}
First, fix $\theta_0\in\Theta$. By Lemma~\ref{lem:ABD-equivalence},
\begin{equation}\label{eq:app-generalized-CE-bracket}
 \overrightarrow{\underline V_{\mathrm{LAM}}(R)}\succsim^r_{\theta_0}R
 \succsim^r_{\theta_0}\overrightarrow{\overline V_{\mathrm{LAM}}(R)} \quad \forall R\in \mathcal{R}_H
\end{equation}
Next, we construct certainty equivalents using Constant Mixture
Continuity. Fix $R \in \mathcal{R}_H$. For each $\beta\in[0,1]$, define $c_\beta(R)=(1-\beta)\underline V_{\mathrm{LAM}}(R)+\beta \overline V_{\mathrm{LAM}}(R)$ and define
\[
 U_R=\{\beta:\overrightarrow{c_\beta(R)}
                    \succsim^r_{\theta_0}R\}
 \quad
 L_R=\{\beta:R\succsim^r_{\theta_0}
                    \overrightarrow{c_\beta(R)}\}
\]
These are closed subsets of $[0,1]$ by Constant Mixture Continuity.
Their union is $[0,1]$ by completeness, and \eqref{eq:app-generalized-CE-bracket} gives $0\in U_R$ and $1\in L_R$.
Connectedness of $[0,1]$ therefore implies $U_R\cap L_R\neq\varnothing$.
Consequently, there exists $V_{\theta_0}(R)\in[\underline V_{\mathrm{LAM}}(R),\overline V_{\mathrm{LAM}}(R)]$ such that $R\sim^r_{\theta_0}\overrightarrow{V_{\theta_0}(R)}$. By Transitivity and Constant Calibration, $V_{\theta_0}(R)$ is unique and for every $R,R'\in\mathcal R_H$,
\[
 R\succsim^r_{\theta_0}R'
 \iff V_{\theta_0}(R)\leq V_{\theta_0}(R')
\]
In particular, $V_{\theta_0}$ is a normalized risk representation. Define
\[
 \alpha^r_{\theta_0}(R)=
 \begin{cases}
 \displaystyle\frac{\overline V_{\mathrm{LAM}}(R)-V_{\theta_0}(R)}{\overline V_{\mathrm{LAM}}(R)-\underline V_{\mathrm{LAM}}(R)}
       &\underline V_{\mathrm{LAM}}(R)<\overline V_{\mathrm{LAM}}(R),\\[6pt]
 0    &\underline V_{\mathrm{LAM}}(R)=\overline V_{\mathrm{LAM}}(R)
 \end{cases}
\]
The certainty-equivalent bounds imply
$\alpha^r_{\theta_0}(R)\in[0,1]$, and by construction
\[
 V_{\theta_0}(R)
 =\alpha^r_{\theta_0}(R)\underline V_{\mathrm{LAM}}(R)
  +(1-\alpha^r_{\theta_0}(R))\overline V_{\mathrm{LAM}}(R)
\]
Define $a_{\theta_0}(F):=\alpha^r_{\theta_0}(-F)$ for $F\in\mathcal F_H$. Then, $U_{a_{\theta_0}}(F)=-V_{\theta_0}(-F)$ represents $\succsim_{\theta_0}$ on
$\mathcal F_H$. By Lemma~\ref{lem:F_H}, $F\mapsto U_{a_{\theta_0}}(F_{\theta_0})$ represents $\succsim_{\theta_0}$ on $\mathcal F$. By Lemma~\ref{lem:common-generalized-selector}, we may choose
$a:\mathcal F_H\to[0,1]$ which represents each $\succsim_{\theta_0}$ on $\mathcal F$. Finally, $\Theta$-Unanimity and Lemma~\ref{lem:unanimity-char} imply
\[
 F\succsim G
 \iff U_a(F_{\theta_0})\geq U_a(G_{\theta_0})
 \quad\text{for every }\theta_0\in\Theta
\]
as desired.

\underline{Backwards direction.}
Suppose that $a$ gives the representation in part (ii). Lemmas~\ref{lem:unanimity-char} and \ref{lem:F_H} yield
$\Theta$-Unanimity and the representation $U_a$ of each
$\succsim_{\theta_0}$ on $\mathcal F_H$.
Fix $\theta_0$ and define $V(R):=-U_a(-R)$. Weak Order, Constant Calibration, and Constant Mixture
Continuity are immediate. Lemma~\ref{lem:generalized-alpha-static-I-representation} shows that for each $I\in \mathcal{I}$, $\succsim_{\theta_0,I}$ is represented on $\mathbb{R}_+^H$ (and hence $\mathbb{R}_+^I$) by $r_I\mapsto\max_{h\in I}r_I(h)$. The GS89 axioms and $I$-Maximal Caution follow from Lemma~\ref{lem-gs}.
Moreover, $\underline V_{\mathrm{LAM}}(R)\leq V(R)\leq \overline V_{\mathrm{LAM}}(R)$ gives
\eqref{eq:app-generalized-CE-bracket}, which by Lemma~\ref{lem:ABD-equivalence} yields Asymptotic Consistency and Caution. This proves part (i).
\end{proof}

\end{document}